\documentclass[11pt,draftcls,onecolumn]{IEEEtran}
\usepackage{graphicx}
\usepackage{epsfig}
\usepackage{amsmath}
\usepackage{mathrsfs}
\usepackage{amssymb}
\usepackage{float}
\usepackage{url}
\usepackage{verbatim}
\usepackage{dsfont}
\usepackage{enumerate}

\usepackage[bottom=0.8in]{geometry}

\newtheorem{theorem}{Theorem}
\newtheorem{lemma}{Lemma}
\newtheorem{proposition}{Proposition}
\newtheorem{corollary}{Corollary}
\newtheorem{question}{Question}
\newtheorem{definition}{Definition}

\newtheorem{remark}{Remark}
\newcommand{\expectation}{\ensuremath{\mathbb{E}}}

\newcommand{\set}{\ensuremath{\mathcal}}
\DeclareMathOperator*{\esssup}{ess\,sup}

\newcommand{\dint}{\displaystyle{\int}}
\newcommand{\overbar}[1]{\mkern 1.5mu\overline{\mkern-1.5mu#1\mkern-1.5mu}\mkern 1.5mu}

\begin{document}
\title{\huge{On the R\'{e}nyi Divergence, Joint Range of Relative Entropies,
and a Channel Coding Theorem}}

\author{\vspace*{0.5cm} Igal Sason
\thanks{
I. Sason is  with the Department of Electrical Engineering, Technion--Israel
Institute of Technology, Haifa 32000, Israel (e-mail: sason@ee.technion.ac.il).
This work has been supported by the Israeli Science Foundation, grant~12/12.}
\thanks{The paper has been submitted to the IEEE Transactions on Information
Theory in February 17, 2015.}
\thanks{This work has been presented at the 2015 IEEE International Symposium
on Information Theory, Hong Kong, June 14--19, 2015.}}

\maketitle

\thispagestyle{empty}

\begin{abstract}
This paper starts by considering the minimization of the R\'{e}nyi divergence
subject to a constraint on the total variation distance. Based on the solution
of this optimization problem, the exact locus of the points
$\bigl( D(Q\|P_1), D(Q\|P_2) \bigr)$ is determined when $P_1, P_2, Q$ are arbitrary
probability measures which are mutually absolutely continuous, and the total variation
distance between $P_1$ and $P_2$ is not below a given value. It is further shown that
all the points of this convex region are attained by probability measures which are
defined on a binary alphabet. This characterization yields a geometric interpretation
of the minimal Chernoff information subject to a constraint on the variational distance.
\par This paper also derives an exponential upper bound on the performance of binary linear
block codes (or code ensembles) under maximum-likelihood decoding.
Its derivation relies on the Gallager bounding technique, and it reproduces
the Shulman-Feder bound as a special case. The bound is expressed in terms of the R\'{e}nyi
divergence from the normalized distance spectrum of the code (or the average distance spectrum of
the ensemble) to the binomially distributed distance spectrum of the capacity-achieving
ensemble of random block codes.
This exponential bound provides a quantitative measure of the degradation in performance of
binary linear block codes (or code ensembles) as a function of the deviation of their
distance spectra from the binomial distribution. An efficient use of this bound
is considered.
\end{abstract}


\section{Introduction}
\label{section: Introduction}
The R\'{e}nyi divergence, introduced in \cite{Renyientropy}, has been studied so
far in various information-theoretic contexts (and it has been actually used before
it had a name \cite{SGB}).
These include generalized cutoff rates and error exponents for hypothesis testing
(\cite{AlajajiCR, Csiszar95, Shayevitz_ISIT11}), guessing moments
(\cite{AtarM15, ErvenH_ISIT10}), source and channel coding error exponents
(\cite{AtarM15, Gallager_book1968, Mansour09, PolyanskiyV10, SGB}),
strong converse theorems for classes of networks \cite{FongT14},
strong data processing theorems for discrete memoryless channels
\cite{Raginsky13}, bounds for joint source-channel coding \cite{TridenskiZI15},
and one-shot bounds for information-theoretic problems \cite{Warsi13}.

In \cite{Gilardoni10}, Gilardoni derived a Pinsker-type lower bound on the R\'enyi
divergence $D_\alpha(P \| Q)$ for $\alpha \in (0,1)$. In view of the fact that this
lower bound is not tight, especially when the total variation distance $|P-Q|$ is large, this
paper starts by considering the minimization of the R\'{e}nyi divergence $D_\alpha(P \| Q)$,
for an arbitrary $\alpha>0$, subject to a given (or minimal) value of the total variation
distance. Note that the minimization here is taken over all probability measures
with a total variation distance which is not below a given value; this problem differs
from the type of problems studied in \cite{BerendHK14} and \cite{OrdentlichW05},
in connection to the minimization of the relative entropy $D(P\|Q)$ subject to a
minimal value of the total variation distance with a fixed probability measure $Q$.
The solution of this problem generalizes the problem of minimizing the relative entropy
$D(P \| Q)$ subject to a given value of the total variation distance where the latter
is a special case with $\alpha = 1$ (see \cite{FedotovHT_IT03, Gilardoni06, ReidW11}).

One possible way to deal with this problem stems from the fact that the R\'enyi divergence
is a one-to-one transformation of the Hellinger divergence $\mathscr{H}_\alpha(P \| Q)$ where
for $\alpha \in (0,1) \cup (1,\infty)$:
\begin{align}\label{renyimeetshellinger}
D_\alpha(P \| Q ) = \frac1{\alpha -1} \; \log \left( 1 + (\alpha - 1)
\, \mathscr{H}_\alpha(P \| Q) \right)
\end{align}
and $\mathscr{H}_\alpha(P \| Q)$ is an $f$-divergence; since the total variation
distance is also an $f$-divergence, this problem can be viewed as a minimization
of an $f$-divergence subject to a constraint on another $f$-divergence.
The numerical optimization of an $f$-divergence subject to simultaneous
constraints on $f_i$-divergences $(i=1, \ldots , L)$ was recently studied
in \cite{GSS_IT14}, where it has been shown that it suffices to restrict
attention to alphabets of cardinality $L+2$.
In fact, as shown in \cite[(22)]{Vajda_1972}, a binary alphabet suffices
if there is a single constraint (i.e., $L=1$) which is on the total variation
distance. In view of \eqref{renyimeetshellinger}, the same conclusion also
holds when minimizing the R\'enyi divergence subject to a constraint on the
total variation distance. To set notation, the divergences $D(P \| Q),
\, |P-Q|, \, \mathscr{H}_\alpha(P \| Q), \, D_\alpha(P \| Q)$ are defined at
the end of this section, being consistent with the notation in \cite{ISSV15-full paper}
and \cite{Verdu_book}.

This paper treats this minimization problem of the R\'{e}nyi divergence in a
different way. We first generalize the analysis in \cite{FedotovHT_IT03},
which was used for the minimization of the relative entropy subject to a
constraint on the variational distance, for proving that
it suffices to restrict attention to probability measures which are defined on
a binary alphabet. Furthermore, the continuation of the analysis in this paper
relies on the Lagrange duality, and a solution of the Karush-Kuhn-Tucker (KKT)
equations while asserting strong duality for the studied problem. The use of
Lagrange duality further simplifies the computational task of the studied
minimization problem.

As complementary results to the minimization problem studied in this paper,
the reader is referred to \cite[Section~8]{ISSV15-full paper} which provides
upper bounds on the R\'{e}nyi divergence $D_\alpha(P \| Q)$ for an arbitrary
$\alpha \in (0,\infty)$ as a function of either the total variation
distance or relative entropy in case that the relative information is bounded.

The solution of the minimization problem of the R\'{e}nyi divergence,
subject to a constraint on the total variation distance, provides an elegant
way for the characterization of the exact locus of the points
$\bigl( D(Q\|P_1), D(Q\|P_2) \bigr)$ where $P_1$ and $P_2$ are probability
measures whose total variation distance is not below a given value $\varepsilon$,
and $Q$ is an arbitrary probability measure. It is further shown in this paper
that all the points of this convex region can be attained by a triple of probability
measures $(P_1, P_2, Q)$ which are defined on a binary alphabet.

In view of the characterization of the exact locus of these points, a geometric
interpretation is provided in this paper for the minimal Chernoff information
between $P_1$ and $P_2$, denoted by $C(P_1, P_2)$, subject to an $\varepsilon$-separation
constraint on the variational distance between $P_1$ and $P_2$.
It is demonstrated in the following that the intersection point at the boundary of the
locus of $\bigl( D(Q\|P_1), D(Q\|P_2) \bigr)$ and the straight line
$D(Q\|P_1) = D(Q\|P_2)$ is the point whose coordinates are equal to the
minimal value of $C(P_1, P_2)$ under the constraint $|P_1 - P_2| \geq \varepsilon$.
The reader is referred to \cite{YardiKV14}, which relies on the closed-form expression
in \cite[Proposition~2]{Sason_IT15} for the minimization of the constrained Chernoff
information, and which analyzes the problem of channel-code detection by a third-party
receiver via the likelihood ratio test. In the latter problem, a third-party receiver
has to detect the channel code used by the transmitter by observing a large number of
noise-affected codewords; this setup has applications in security or cognitive radios,
or in link adaptation in some wireless technologies.

Since the R\'{e}nyi divergence $D_\alpha(P\|Q)$ forms a generalization of the
relative entropy $D(P\|Q)$, where the latter corresponds to $\alpha=1$, the
approach suggested in this paper for the characterization of the exact locus
of pairs of relative entropies in view of a solution to a minimization problem
of the R\'{e}nyi divergence is analogous to the usefulness of complex analysis
in solving real-valued problems.
We consider the analysis of the considered problem as mathematically pleasing
in its own right. Note, however, that an operational meaning of a special point
at the boundary of this locus has an operational meaning in view of
\cite{YardiKV14} (see the previous paragraph).
The studied problem considered here differs from the study in \cite{HarremoesV_2011}
which considered the joint range of $f$-divergences for pairs (rather than triplets)
of probability measures.

The performance analysis of linear codes under maximum-likelihood (ML) decoding is
of interest for studying the potential performance of these codes under optimal
decoding, and for the evaluation of the degradation in performance that is incurred by
the use of sub-optimal and practical decoding algorithms. The reader is referred
to \cite{Sason_Shamai_FnT} which is focused on this topic.

The second part of this paper derives an exponential upper bound on the performance
of ML decoded binary linear block codes (or code ensembles). Its derivation relies
on the Gallager bounding technique (see \cite[Chapter~4]{Sason_Shamai_FnT},
\cite{Shamai_Sason_IT02}), and it reproduces the Shulman-Feder bound \cite{Shulman_Feder99}
as a special case. The new exponential bound derived in this paper is expressed in
terms of the R\'{e}nyi divergence from the normalized distance spectrum of the code
(or average distance spectrum of the ensemble) to the binomial distribution which
characterizes the average distance spectrum of the capacity-achieving ensemble of fully
random block codes. This exponential bound provides a quantitative measure of the degradation
in performance of binary linear block codes (or code ensembles) as a function of the
deviation of their (average) distance spectra from the binomial distribution, and its
use is exemplified for an ensemble of turbo-block codes.

This paper is structured as follows:
Section~\ref{section: Minimum of the Renyi divergence subject to a fixed TV distance}
solves the minimization problem for the R\'{e}nyi divergence under a constraint on the total variation
distance, Section~\ref{section: Range of relative entropies subject to a minimal TV distance} uses
the solution of this minimization problem to obtain an exact characterization of the joint range of
the relative entropies in the considered setting above.
Section~\ref{section: An Upper Bound on the ML Decoding Error Probability with the Renyi Divergence}
provides a new exponential upper bound on the block error probability of ML decoded binary linear block codes,
which is expressed in terms of the R\'{e}nyi divergence, suggests an efficient way to apply the bound to
the performance evaluation of binary linear block codes (or code ensembles), and exemplifies its use.
Throughout this paper, logarithms are to the base~$e$.

We end this section by introducing the definitions and notation used in this work,
which are consistent with \cite{ISSV15-full paper}, \cite{Verdu_book}, and are included
here for the convenience of the reader.

\subsection*{Definitions and Notation}

We assume throughout that the probability measures $P$ and $Q$ are
defined on a common measurable space $(\set{A}, \mathscr{F})$, and $P \ll Q$
denotes that $P$ is {\em absolutely continuous} with respect to $Q$, namely
there is no event $\set{F} \in \mathscr{F}$ such that $P(\set{F}) > 0 = Q(\set{F})$.
Let $\frac{\text{d}P}{\text{d}Q}$ denote the Radon-Nikodym derivative (or density)
of $P$ with respect to $Q$.

\vspace*{0.1cm}
\begin{definition}[Relative entropy]
The relative entropy is given by
\begin{align} \label{eq: relative entropy}
D(P\|Q) &= \int_{\set{A}} \mathrm{d}P \, \log \left(\frac{\mathrm{d}P}{\mathrm{d}Q}\right).
\end{align}
\end{definition}

\vspace*{0.1cm}
\begin{definition}[Total variation distance]
The total variation distance is given by
\begin{align}
\label{eq1: TV distance}
|P-Q| &=  \int_{\set{A}} \left|\frac{\text{d}P}{\text{d}Q} - 1 \right| \, \text{d}Q.
\end{align}
\end{definition}

\vspace*{0.1cm}
\begin{definition}[Hellinger divergence]
The Hellinger divergence of order $\alpha \in (0,1) \cup (1, \infty)$
is given by
\begin{align}
\mathscr{H}_{\alpha}(P \| Q) &= \frac{1}{\alpha-1} \left( \int_{\set{A}} \mathrm{d}Q \left(\frac{\mathrm{d}P}{\mathrm{d}Q}\right)^\alpha - 1 \right).
\end{align}
The analytic extension of $\mathscr{H}_{\alpha}(P \| Q)$ at $\alpha=1$ yields
$\mathscr{H}_1(P \| Q) = D(P \| Q)$ (nats).
\end{definition}

\vspace*{0.1cm}
\begin{definition}[R\'{e}nyi divergence]
The {\em R\'{e}nyi divergence of order $\alpha \geq 0$} is given as follows:
\begin{itemize}
\item
If $\alpha \in (0,1) \cup (1, \infty) $, then
\begin{align}
\label{eq:RD1}
D_{\alpha}(P\|Q) &= \frac1{\alpha-1} \; \log \left( \int_{\set{A}} \mathrm{d}Q \left(\frac{\mathrm{d}P}{\mathrm{d}Q}\right)^\alpha \right).
\end{align}
\item If $\alpha = 0$,  then
\begin{align} \label{eq: d0}
D_0 (P \| Q ) = \max_{\set{F} \in \mathscr{F}\colon P(\set{F}) = 1} \log \left(\frac1{Q (\set{F})}\right).
\end{align}
\item $D_1(P\|Q) = D(P\|Q)$ which is the analytic extension of $D_{\alpha}(P \| Q)$ at $\alpha=1$.
\item If $\alpha = +\infty$ then
\begin{align} \label{def:dinf}
D_{\infty}(P\|Q) = \log \left(\esssup \frac{\text{d}P}{\text{d}Q} \, (Y)\right)
\end{align}
with $Y \sim Q$.
\end{itemize}
\end{definition}

\begin{definition}[Chernoff information]
The Chernoff information between probability measures $P_1$ and $P_2$ is expressed as follows in
terms of the R\'{e}nyi divergence:
\begin{align}
C(P_1, P_2) = \max_{\alpha \in [0,1]} \Bigl\{(1-\alpha) D_{\alpha}(P_1\|P_2)\Bigr\}
\end{align}
and it is the best achievable exponent in the Bayesian probability of error for binary hypothesis testing
(see, e.g., \cite[Theorem~11.9.1]{Cover_Thomas}).
\end{definition}

\section{Minimization of the R\'{e}nyi Divergence with a Constrained Total Variation Distance}
\label{section: Minimum of the Renyi divergence subject to a fixed TV distance}

In this section, we derive a tight lower bound on the R\'{e}nyi divergence $D_{\alpha}(P_1 \| P_2)$
subject to an equality constraint on the total variation distance $|P_1 - P_2| = \varepsilon$ where
$\varepsilon \in [0,2)$ is fixed; alternatively, it can regarded as a minimization problem under the
inequality constraint $|P_1 - P_2| \geq \varepsilon$.
It is first shown that this lower bound is attained for probability measures defined on a binary alphabet,
and Lagrange duality is used to further simplify the computational task of this bound.
The special case where $\alpha=1$, which is specialized to the minimization of the relative entropy
subject to a fixed total variation distance, has been studied extensively, and three equivalent
forms of the solution to this optimization problem were derived in \cite{FedotovHT_IT03},
\cite{Gilardoni06}, \cite{ReidW11}.

In \cite[Corollaries~6 and 9]{Gilardoni10}, Gilardoni derived two Pinsker-type lower bounds
on the R\'{e}nyi divergence of order $\alpha \in (0,1)$, expressed in terms of the total
variation distance. Among these two bounds, the improved lower bound is given (in nats) by
\begin{equation}
D_{\alpha}(P \| Q) \geq \tfrac12 \, \alpha \varepsilon^2 + \tfrac19 \, \alpha (1+5\alpha-5\alpha^2)
\varepsilon^4, \quad \forall \, \alpha \in (0,1)
\label{eq: Gilardoni's lower bound}
\end{equation}
where $|P-Q|=\varepsilon$ denotes the total variation distance between $P$ and $Q$. Note that in the limit
where $\varepsilon$ tend to~2 (from below), this lower bound converges to a finite value which is at most
$\frac{22}{9}$; it is, however, an artifact of the lower bound in view of the next lemma.
\begin{lemma}
\begin{equation}
\lim_{\varepsilon \uparrow 2} \inf_{P, Q \colon |P-Q| =
\varepsilon} D_{\alpha}(P\|Q) = \infty, \quad \forall \, \alpha > 0.
\label{eq: limit when the TV distance tends to 1}
\end{equation}
\label{lemma: limit when the TV distance tends to 1}
\end{lemma}
\begin{proof}
See Appendix~\ref{Appendix: proof of Lemma 1}.
\end{proof}

In the following, we derive a tight lower bound which is shown to be
achievable by a restriction of the probability measures to a binary alphabet.
For $\alpha > 0$, let
\begin{align}
\label{eq: g function - infinimum of Renyi divergence subject to TV distance}
g_{\alpha}(\varepsilon)
& \triangleq \min_{P_1, P_2 \colon |P_1 - P_2| = \varepsilon}
D_{\alpha}(P_1 \| P_2), \\
&= \min_{P_1, P_2 \colon |P_1 - P_2| \geq \varepsilon}
D_{\alpha}(P_1 \| P_2), \quad \forall \, \varepsilon \in [0,2).
\label{eq: g function - 2nd infinimum of Renyi divergence subject to TV distance}
\end{align}

In the following, we evaluate the function $g_{\alpha}$.
In view of \cite[Section~2]{FedotovHT_IT03} which characterizes the minimum
of the relative entropy in terms of the total variation distance, we first
extend the argument in \cite{FedotovHT_IT03} to prove the next lemma.

\begin{lemma}
For an arbitrary $\alpha > 0$, the minimization in \eqref{eq: g function - infinimum of Renyi divergence subject to TV distance}
is attained by probability measures which are defined on a binary alphabet.
\label{lemma: restriction of the minimization to 2-element probability distributions}
\end{lemma}
\begin{proof}
See Appendix~\ref{Appendix: proof of Lemma 2}.
\end{proof}

The following proposition enables to calculate $g_{\alpha}$ for an
arbitrary positive $\alpha$.
\begin{proposition}
Let $\alpha \in (0,1) \cup (1, \infty)$ and $\varepsilon \in [0,2)$.
The function $g_{\alpha}$ in
\eqref{eq: g function - infinimum of Renyi divergence subject to TV distance}
satisfies
\begin{equation}
g_{\alpha}(\varepsilon) = \min_{p, q \in [0,1] \colon |p-q| \geq \frac{\varepsilon}{2}} d_\alpha (p \| q )
\label{eq: optimization problem 1 for the Renyi divergence}
\end{equation}
where
\begin{equation}
d_\alpha (p \| q ) \triangleq \frac{\log \Bigl(p^{\alpha} q^{1-\alpha} + (1-p)^{\alpha}
(1-q)^{1-\alpha} \Bigr)}{\alpha-1}
\label{eq: binary Renyi divergence}
\end{equation}
denotes the binary R\'{e}nyi divergence.
\label{proposition: optimization problem 1 for the Renyi divergence}
\end{proposition}
\begin{proof}
This directly follows from Lemma~\ref{lemma: restriction of the minimization to 2-element probability distributions}.
\end{proof}

\begin{proposition}
\begin{align}
g_{\frac{1}{2}}(\varepsilon) = -\log\bigl(1-\tfrac14 \, \varepsilon^2\bigr),
\quad \forall \, \varepsilon \in [0,2)
\label{eq: g for alpha=1/2}
\end{align}
and
\begin{align}
g_2(\varepsilon) = \left\{
\begin{array}{ll}
\log(1+\varepsilon^2), & \quad \mbox{if $\varepsilon \in [0,1]$,} \\[0.1cm]
-\log \left(1-\tfrac12 \, \varepsilon \right), & \quad
\mbox{if $\varepsilon \in (1, 2).$}
\end{array}
\right.
\label{eq: g for alpha=2}
\end{align}
Furthermore, for $\alpha \in (0,1)$ and $\varepsilon \in [0,2)$,
\begin{equation}
g_{\alpha}(\varepsilon) =
\left(\frac{\alpha}{1-\alpha}\right)
g_{1-\alpha}(\varepsilon),
\label{eq: skew-symmetry property of g}
\end{equation}
and
\begin{equation}
g_{\alpha}(\varepsilon) \geq c_1(\alpha) \,
\log\left(\frac{1}{1- \tfrac12 \, \varepsilon}\right)+c_2(\alpha),
\label{eq: lower bound on g}
\end{equation}
where
\begin{align}
c_1(\alpha) \triangleq \min\left\{1, \frac{\alpha}{1-\alpha}
\right\}, \quad c_2(\alpha) \triangleq -\frac{\log \, 2}{1-\alpha}.
\label{eq: expressions for c_1, c_2}
\end{align}
\label{proposition: skew-symmetry property of g}
\end{proposition}
\begin{proof}
See Appendix~\ref{Appendix: proof of proposition on skew-symmetry of g}.
\end{proof}

\begin{remark}
The lower bound on $g_{\alpha}(\cdot)$ in \eqref{eq: lower bound on g} provides another
proof of Lemma~\ref{lemma: limit when the TV distance tends to 1} since it first yields that
$\lim_{\varepsilon \uparrow 2} g_{\alpha}(\varepsilon) = \infty$
for $\alpha \in (0,1)$; this lemma also holds for $\alpha \geq 1$
since $D_{\alpha}(P\|Q)$ is monotonically increasing in its order $\alpha$.
\end{remark}

\par
In the following, we use Lagrange duality to obtain an alternative
expression as a solution of the minimization problem for $g_{\alpha}$.
Recall that Proposition~\ref{proposition: optimization problem 1 for the Renyi divergence}
applies to every $\alpha > 0$.
The following enables to simplify considerably the computational task in calculating $g_{\alpha}$,
for $\alpha \in (0,1)$.

\begin{lemma}
Let $\alpha \in (0,1)$ and $\varepsilon' \in (0,1)$. The function
\begin{equation}
f_{\alpha, \varepsilon'}(q) \triangleq
\frac{\left(1-\frac{\varepsilon'}{1-q}\right)^{\alpha-1} -
\left(1+\frac{\varepsilon'}{q}\right)^{\alpha-1}}{\left(1+
\frac{\varepsilon'}{q}\right)^\alpha-
\left(1-\frac{\varepsilon'}{1-q}\right)^\alpha} \, , \, \quad
\forall q \in (0, 1-\varepsilon')
\label{eq: f of the 2 parameters alpha and epsilon}
\end{equation}
is strictly monotonically increasing, positive, continuous, and
\begin{align}
\lim_{q \rightarrow 0^+} f_{\alpha, \varepsilon'}(q) = 0, \qquad
\lim_{q \rightarrow (1-\varepsilon')^-} f_{\alpha, \varepsilon'}(q) = +\infty.
\label{eq: limits of f at the 2 endpoints of the interval}
\end{align}
\label{lemma: monotonicity of f}
\end{lemma}
\begin{proof}
See Appendix~\ref{Appendix: proof of lemma on monotonicity}.
\end{proof}

\begin{corollary}
For $\alpha \in (0,1)$ and $\varepsilon' \in (0,1)$, the equation
\begin{align}
f_{\alpha, \varepsilon'}(q) = \tfrac{1-\alpha}{\alpha}
\label{eq: equation for f of the 2 parameters alpha and epsilon}
\end{align}
has a unique solution $q \in (0, 1-\varepsilon')$.
\end{corollary}
\begin{proof}
It follows from Lemma~\ref{lemma: monotonicity of f},
and the mean value theorem for continuous functions.
\end{proof}

\begin{remark}
Since $f_{\alpha, \varepsilon'} \colon (0, 1-\varepsilon') \rightarrow (0, \infty)$
is strictly monotonically increasing (see
Lemma~\ref{lemma: monotonicity of f}), the numerical calculation of
the unique solution of equation~\eqref{eq: equation for f of the 2 parameters alpha and epsilon}
is easy.
\label{remark: numerical solution of the equation}
\end{remark}

An alternative simplified form for the optimization problem in
Proposition~\ref{proposition: optimization problem 1 for the Renyi divergence}
is next provided for orders $\alpha \in (0,1)$. Hence,
Proposition~\ref{proposition: optimization problem 1 for the Renyi divergence}
applies to every $\alpha>0$, whereas the following is restricted to $\alpha \in (0,1)$.
This, however, proves to be very useful in the next section in terms of obtaining
a significant reduction in the computational complexity of $g_\alpha(\cdot)$ where
only $\alpha \in (0,1)$ is of interest there.\footnote{This saving in the
computational complexity accelerated the running time of the numerical calculations
in our computer by two orders of magnitude.}

\begin{proposition}
Let $\alpha \in (0,1)$, $\varepsilon \in (0,2)$, and let $\varepsilon' = \tfrac{\varepsilon}{2}$.
A solution of the minimization problem for $g_{\alpha}(\varepsilon)$ in
Proposition~\ref{proposition: optimization problem 1 for the Renyi divergence}
is obtained by calculating the binary R\'{e}nyi divergence $d_\alpha (p \| q )$
in \eqref{eq: binary Renyi divergence} while
taking the unique solution $q \in (0, 1- \varepsilon')$ of
\eqref{eq: equation for f of the 2 parameters alpha and epsilon},
and setting $p = q+\varepsilon'$.
\label{proposition: efficient calculation of g_alpha}
\end{proposition}
\begin{proof}
See Appendix~\ref{Appendix: proof of proposition on efficient calculation of g}.
\end{proof}

In view of Proposition~\ref{proposition: efficient calculation of g_alpha},
the plots in Figures~\ref{Figure: Renyi_divergence_alpha_025_050_075_100}
and~\ref{Figure: Renyi_divergence_alpha09} provide numerical results.

\begin{figure}[here!]
\begin{center}
\hspace*{-1cm}
\epsfig{file=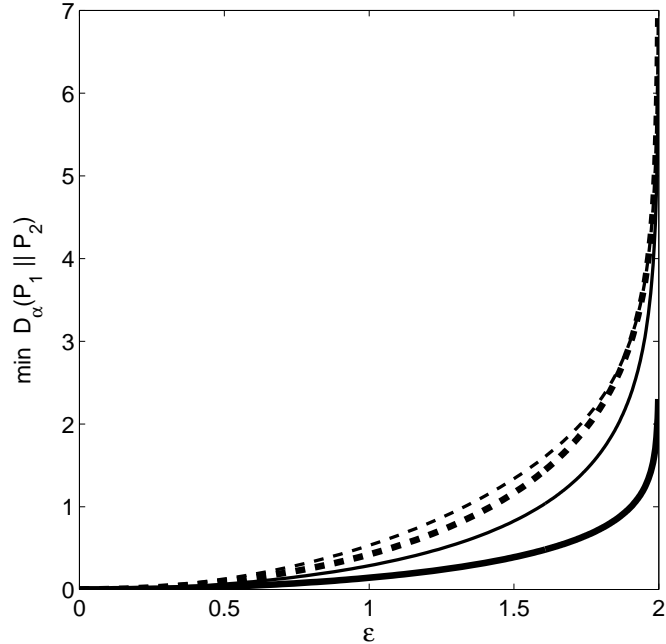,scale=0.62}
\caption{\label{Figure: Renyi_divergence_alpha_025_050_075_100} A plot of the
minimum of the R\'{e}nyi divergence $D_{\alpha}(P_1\|P_2)$ subject to
the constraint $|P_1-P_2| \geq \varepsilon$ where $\varepsilon \in [0,2)$.
The curves in this plot correspond to $\alpha = 0.25$ (thick solid curve),
$\alpha=0.50$ (thin solid curve), $\alpha=0.75$ (thick dashed curve), and
$\alpha=1.00$ (thin dashed curve, referring to the relative entropy).}
\end{center}
\end{figure}

\begin{figure}[here!]
\begin{center}
\hspace*{-1cm}
\epsfig{file=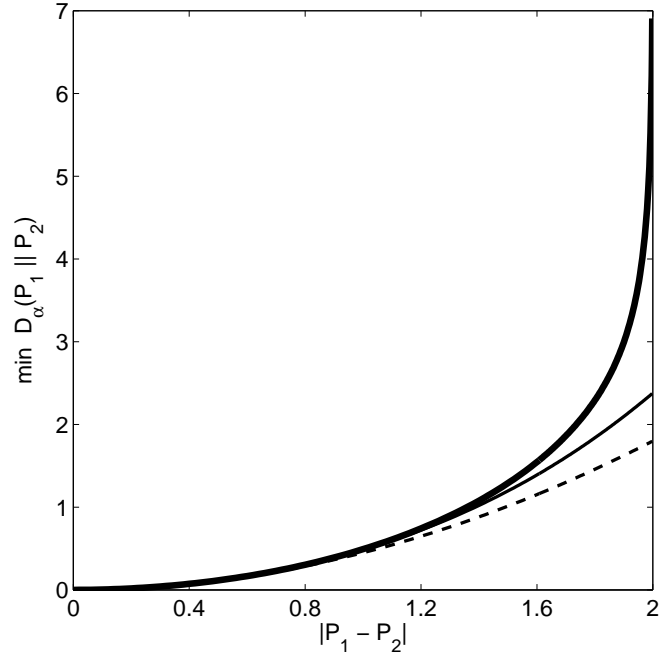,scale=0.62}
\caption{\label{Figure: Renyi_divergence_alpha09} A plot of the
minimum of the R\'{e}nyi divergence $D_{\alpha}(P_1\|P_2)$ of order
$\alpha = 0.90$ subject to the constraint $|P_1-P_2| \geq \varepsilon \in [0,2)$.
The exact minimum (thick solid curve) is compared with the
Pinsker-type lower bound in \cite[Corollary~9]{Gilardoni10}
(the thin solid curve), and its weaker version in
\cite[Corollary~6]{Gilardoni10} (the dashed curve).}
\end{center}
\end{figure}

\newpage

\section{The Locus of $(D(Q\|P_1), D(Q\|P_2))$ With a Constrained
Total Variation Distance}
\label{section: Range of relative entropies subject to a minimal TV distance}

In this section, we address the following question:
\begin{question}
What is the locus of the points $\bigl( D(Q\|P_1), D(Q\|P_2) \bigr)$
if $P_1, P_2, Q$ are arbitrary probability measures which are mutually
absolutely continuous, and $|P_1 - P_2| \geq \varepsilon$ for a given value
$\varepsilon \in (0,2)$ ? (none of the three probability measures is fixed).
\label{question}
\end{question}

The present section provides an exact characterization of this locus
in view of the solution to the minimization problem in
Section~\ref{section: Minimum of the Renyi divergence subject to a fixed TV distance},
and the following lemma:
\begin{lemma}
Let $P_1, P_2, Q$ be pairwise mutually absolutely continuous probability measures
defined on a measurable space $(\set{A}, \mathscr{F})$. Then, for
$\alpha \in (0,1) \cup (1, \infty)$,
\begin{align}
D_{\alpha}(P_1 \| P_2) = D(Q \| P_2) + \tfrac{\alpha}{1-\alpha} \cdot D(Q\|P_1)
+ \tfrac{1}{\alpha-1} \cdot D(Q \| Q_{\alpha})
\label{eq: decomposition of a Renyi divergence into relative entropies}
\end{align}
where the probability measure $Q_{\alpha}$ is given by
\begin{align}
\frac{\mathrm{d}Q_{\alpha}}{\mathrm{d}Q} \, (x) = \frac{\left(\frac{\mathrm{d}P_1}{\mathrm{d}Q} \, (x)\right)^\alpha
\left(\frac{\mathrm{d}P_2}{\mathrm{d}Q} \, (x)\right)^{1-\alpha}}{\dint_{\set{A}}{\left(\frac{\mathrm{d}P_1}{\mathrm{d}Q} \, (u)
\right)^\alpha} \, \left(\frac{\mathrm{d}P_2}{\mathrm{d}Q} \, (u)\right)^{1-\alpha} \, \mathrm{d}Q(u)}, \quad \forall \, x \in \set{A}.
\label{eq: optimized tilted measure}
\end{align}
\label{lemma: identity for the Renyi divergence}
\end{lemma}
\begin{proof}
See Appendix~\ref{Appendix: proof of the lemma with the identity for the Renyi divergence}.
\end{proof}

As a corollary of Lemma~\ref{lemma: identity for the Renyi divergence}, the
following tight inequality holds, which is attributed to van Erven
\cite[Lemma~6.6]{ErvenH10_PhD} and Shayevitz \cite[Section~IV.B.8]{Shayevitz_arXiv10}).
It will be useful for the continuation of this section, jointly with the results of
Section~\ref{section: Minimum of the Renyi divergence subject to a fixed TV distance}.
\begin{corollary}
Let $P_1 \ll \gg P_2$ be mutually absolutely continuous discrete probability measures
defined on a common set $\set{A}$. If $\alpha \in (0,1)$ then
\begin{align}
\tfrac{\alpha}{1-\alpha} \cdot D(Q\|P_1) + D(Q\|P_2) \geq D_{\alpha}(P_1 \| P_2)
\label{Shayevitz inequality}
\end{align}
with equality if and only if, for every $x \in \set{A}$,
\begin{align}
Q(x) = \frac{P_1(x)^\alpha \, P_2(x)^{1-\alpha}}{\sum_{u \in \set{A}} P_1(u)^\alpha \, P_2(u)^{1-\alpha}}.
\label{Shayevitz function}
\end{align}
For $\alpha > 1$, inequality~\eqref{Shayevitz inequality} is reversed with the same necessary
and sufficient condition for an equality.
\label{corollary: Shayevitz}
\end{corollary}

\begin{remark}
The knowledge of the maximizing probability measure in \eqref{Shayevitz function}
is required for the characterization of the exact locus which is studied in this section.
\end{remark}

\vspace*{0.1cm}
The exact locus of the points $\bigl(D(Q\|P_1), D(Q\|P_2) \bigr)$
is determined as follows: let $|P_1-P_2| \geq \varepsilon$ for a fixed
$\varepsilon \in (0,2)$, and let $\alpha \in (0,1)$ be chosen arbitrarily. By the tight lower
bound in Section~\ref{section: Minimum of the Renyi divergence subject to a fixed TV distance},
we have
\begin{equation}
D_{\alpha}(P_1 \| P_2) \geq g_{\alpha}(\varepsilon)
\label{eq: tight lower bound on Renyi divergence}
\end{equation}
where $g_{\alpha}$ is expressed in \eqref{eq: optimization problem 1 for the Renyi divergence}.
For $\alpha \in (0,1)$ and for a fixed value of $\varepsilon \in (0,2)$,
let $p = p^\star$ and $q = q^\star$ in $(0,1)$ be set to achieve the global
minimum in \eqref{eq: optimization problem 1 for the Renyi divergence} (note that, without loss
of generality, one can assume that $p \geq q$ since if $(p,q)$ achieves the minimum in
\eqref{eq: optimization problem 1 for the Renyi divergence} then also $(1-p, 1-q)$ achieves the same
minimum). Consequently, the lower bound in \eqref{eq: tight lower bound on Renyi divergence} is
attained by probability measures $P_1, P_2$ which are defined on a binary alphabet (see
Lemma~\ref{lemma: restriction of the minimization to 2-element probability distributions}) with
\begin{align}
\label{eq: pair of 2-element PDs}
\begin{split}
& P_1(0) = p^\star = p^\star(\alpha, \varepsilon), \quad P_1(1) = 1-p^\star; \\
& P_2(0) = q^\star = q^\star(\alpha, \varepsilon), \quad P_2(1) = 1-q^\star.
\end{split}
\end{align}
From Corollary~\ref{corollary: Shayevitz} and \eqref{eq: tight lower bound on Renyi divergence},
\eqref{eq: pair of 2-element PDs},
it follows that for every $\alpha \in (0,1)$
\begin{equation}
g_{\alpha}(\varepsilon) \leq D(Q \| P_2) + \tfrac{\alpha}{1-\alpha} \cdot D(Q \| P_1)
\label{eq: an inequality characterizing the achievable region}
\end{equation}
where equality in \eqref{eq: an inequality characterizing the achievable region} holds
if $P_1$ and $P_2$ are the probability measures in
\eqref{eq: pair of 2-element PDs} which are defined on a binary alphabet,
and $Q$ is the respective probability measure in
\eqref{Shayevitz function} which is therefore also defined on a binary alphabet.
Hence, there exists a triple of probability measures
$P_1, P_2, Q$ which are defined on a binary alphabet and satisfy
\eqref{eq: an inequality characterizing the achievable region}
with equality, and these probability measures are easy to calculate
for every $\alpha \in (0,1)$ and $\varepsilon \in (0,2)$.

\begin{remark}
Similarly to \eqref{eq: an inequality characterizing the achievable region}, since
$|P_1-P_2| = |P_2-P_1|$, it follows from
\eqref{eq: an inequality characterizing the achievable region} that
\begin{equation}
g_{\alpha}(\varepsilon) \leq D(Q \| P_1) + \tfrac{\alpha}{1-\alpha} \cdot D(Q \| P_2).
\label{eq: 2nd inequality characterizing the achievable region}
\end{equation}
By multiplying both sides of \eqref{eq: 2nd inequality characterizing the achievable region}
by $\tfrac{1-\alpha}{\alpha}$ and relying on the skew-symmetry property in \eqref{eq: skew-symmetry property of g},
it follows that \eqref{eq: 2nd inequality characterizing the achievable region} is equivalent to
$$ g_{1-\alpha}(\varepsilon) \leq D(Q \| P_2) + \tfrac{1-\alpha}{\alpha} \cdot D(Q \| P_1)$$
which is \eqref{eq: an inequality characterizing the achievable region} when $\alpha \in (0,1)$
is replaced by $1-\alpha$. Hence, since \eqref{eq: an inequality characterizing the achievable region} holds
for every $\alpha \in (0,1)$, there is no additional information in
\eqref{eq: 2nd inequality characterizing the achievable region}.
\end{remark}

\begin{theorem}
The exact locus of $\bigl( D(Q \| P_1), D(Q \| P_2) \bigr)$ in the setting
of Question~1 is the convex region whose boundary is the convex envelope of all the straight lines
\begin{align}
D(Q \| P_2) + \tfrac{\alpha}{1-\alpha} \cdot D(Q\|P_1) = g_{\alpha}(\varepsilon), \quad \forall \, \alpha \in (0,1)
\label{eq: straight lines}
\end{align}
(i.e., the boundary is the pointwise maximum of the set of straight lines in
\eqref{eq: straight lines} for $\alpha \in (0,1)$).
Furthermore, all the points in this convex region, including its boundary, are attained by
probability measures $P_1, P_2, Q$ which are defined on a binary alphabet.
\label{theorem: joint range of the relative entropies}
\end{theorem}

\begin{proof}
Let $P_1, P_2, Q$ be arbitrary probability measures which are mutually absolutely continuous
and satisfy the $\varepsilon$ separation condition for $P_1$ and $P_2$ in total variation.
In view of Corollary~\ref{corollary: Shayevitz} and since by definition
$D_{\alpha}(P_1 \| P_2) \geq g_{\alpha}(\varepsilon)$, it follows that the point
$\bigl( D(Q \| P_1), D(Q \| P_2) \bigr)$ satisfies
\begin{align} \label{eq: inequalities required to hold}
D(Q \| P_2) + \tfrac{\alpha}{1-\alpha} \cdot D(Q\|P_1) \geq g_{\alpha}(\varepsilon)
\end{align}
for every $\alpha \in (0,1)$; this implies that every such a point is either
on or above the convex envelope of the parameterized straight lines in \eqref{eq: straight lines}.

\par
We next prove that a point which is below the convex envelope of the lines in
\eqref{eq: straight lines} cannot be achieved under the constraint $|P_1 - P_2| \geq \varepsilon$.
The reason for this claim is because for such a point $\bigl( D(Q \| P_1), D(Q \| P_2) \bigr)$,
there is some $\alpha \in (0,1)$ for which
\begin{align}
D(Q \| P_2) + \tfrac{\alpha}{1-\alpha} \cdot D(Q\|P_1) < g_{\alpha}(\varepsilon)
\end{align}
Since under the $\varepsilon$ separation condition for $P_1$ and $P_2$ in total variation distance,
$D_{\alpha}(P_1 \| P_2) \geq g_{\alpha}(\varepsilon)$, then for such $\alpha \in (0,1)$,
inequality~\eqref{Shayevitz inequality} is violated; in view of Corollary~\ref{corollary: Shayevitz},
this yields that the point is not achievable under the constraint $|P_1-P_2| \geq \varepsilon$.
As an interim conclusion, it follows that the exact locus of the achievable points is the set of all
points in the plane $\bigl( D(Q \| P_1), D(Q \| P_2) \bigr)$ which are on or above the convex envelope
of the parameterized straight lines in \eqref{eq: straight lines} for $\alpha \in (0,1)$.

\par
The next step aims to show that an arbitrary point which is located at the boundary of this region
can be obtained by a triplet of probability measures $(P_1^\star, P_2^\star, Q^\star)$ which are defined on a binary
alphabet, and satisfy $|P_1^\star-P_2^\star|=\varepsilon$. To that end, note that every point which is on the boundary
of this region is a tangent point to one of the straight lines in \eqref{eq: straight lines} for some
$\alpha \in (0,1)$. Accordingly, the proper probability measures $P_1^\star$, $P_2^\star$ and $Q^*$
can be determined as follows for a given $\varepsilon \in (0,2)$:
\begin{enumerate}[a)]
\item Find the slope $s<0$ of the tangent line at the selected point on the boundary; in view of
\eqref{eq: straight lines}, $s=-\tfrac{\alpha}{1-\alpha}$ yields $\alpha = -\frac{s}{1-s} \in (0,1)$.
\item In view of Proposition~\ref{proposition: efficient calculation of g_alpha}, determine
$p_1^\star, p_2^\star \in (0,1)$ such that $|p_1^\star - p_2^\star| = \tfrac{\varepsilon}{2}$
and $d_{\alpha}(p_1^\star \| p_2^\star) = g_{\alpha}(\varepsilon)$. Consequently, let $P_1^\star$
and $P_2^\star$ be the probability measures which are defined on the binary alphabet with
$P_1^\star(0) = p_1^\star$ and $P_2^\star(0) = p_2^\star$.
\item The respective probability measure $Q^\star = Q^\star_{\alpha}$ is calculated from
\eqref{Shayevitz function}, and it is therefore also defined on the binary alphabet.
\end{enumerate}

\par
Finally, we show that every interior point in the achievable region can be attained
as well by a proper selection of $P_1^\star$, $P_2^\star$ and $Q^\star$ which are
defined on a binary alphabet.
To that end, note that every such interior point is located at the boundary of the
locus of $\bigl( D(Q \| P_1), D(Q \| P_2) \bigr)$ under the constraint $|P_1 - P_2| \geq \overbar{\varepsilon}$
with some $\overbar{\varepsilon} \in (\varepsilon, 2)$; this follows from the fact that
$g_{\alpha}(\cdot)$ is a strictly monotonically increasing and continuous function in
$(0,2)$, which tends to infinity as we let $\varepsilon$ tend to~2 (see
Lemma~\ref{lemma: limit when the TV distance tends to 1}). It therefore follows that the
suitable triplet of probability measures $(P_1^\star, P_2^\star, Q^\star)$ can be
obtained by the same algorithm used for points on the boundary of this region, except
for replacing $\varepsilon$ by the larger value~$\overbar{\varepsilon}$.

This concludes the proof by first characterizing the exact locus of points, and then
demonstrating that every point in this convex region (including its boundary) is attained
by probability measures which are defined on the binary alphabet; the proof is also constructive
in the sense of providing an algorithm to calculate such probability measures $P_1^\star,
P_2^\star, Q^\star$ for an arbitrary point in this closed and convex region.
\end{proof}

\begin{figure}[here!]
\begin{center}
\hspace*{-1.4cm}
\epsfig{file=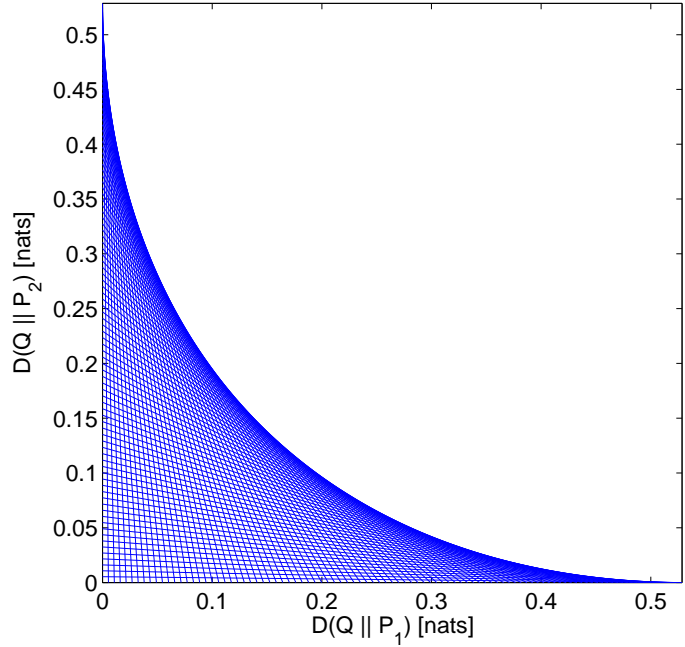,scale=0.62}
\caption{\label{Figure: achievable_region_epsilon1} The exact locus of
$(D(Q\|P_1), D(Q\|P_2))$ where $P_1, P_2$ are arbitrary probability
measures with $|P_1 - P_2| \geq 1$ with $\varepsilon=1$.
The exact locus of these relative entropies includes all the points on
and above the convex envelope of the straight lines in \eqref{eq: straight lines},
which is the convex and closed region painted in white.}
\end{center}
\end{figure}

\begin{figure}[here!]
\begin{center}
\epsfig{file=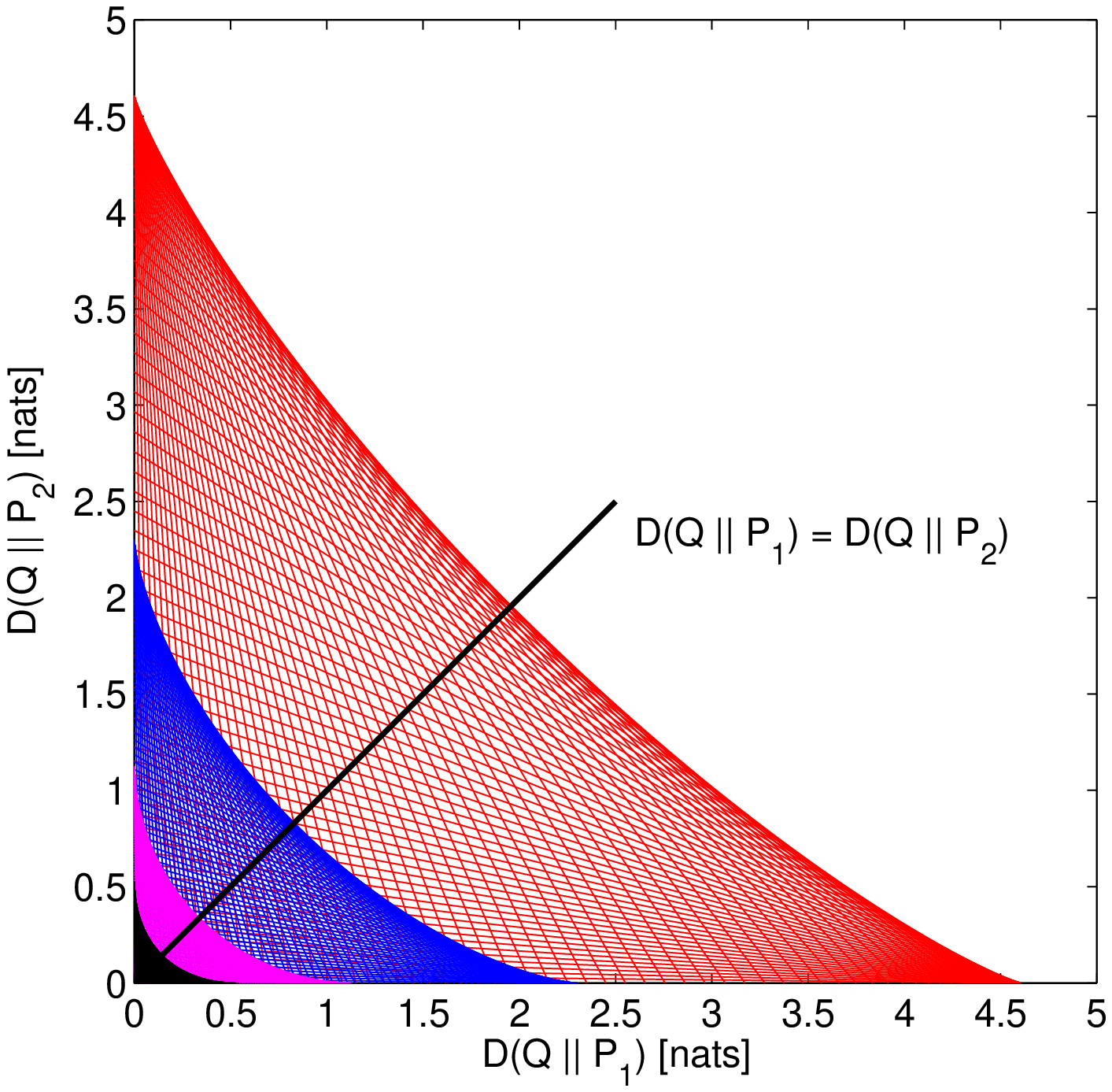,scale=0.62}
\caption{\label{Figure: achievable_regions_epsilon_100_140_180_198}
This plot shows the 4 exact loci of $(D(Q\|P_1), D(Q\|P_2))$ where
$P_1, P_2$ are arbitrary probability measures such that $|P_1 - P_2| \geq \varepsilon$,
with $\varepsilon = 1.00, 1.40, 1.80, 1.98$, and $Q \ll P_1, P_2$ is an arbitrary
probability measure. The exact locus which is above the convex envelope for the
respective value of $\varepsilon$ (painted in white) shrinks as the value of
$\varepsilon$ is increased, especially when $\varepsilon$ is close (from below) to~2.
The intersection of the boundary of the exact locus, for a given $\varepsilon \in [0,2)$,
with the straight line $D(Q\|P_1)=D(Q\|P_2)$ (passing through the origin) is at the point
$\bigl(-\tfrac12 \, \log(1-\varepsilon^2), \, -\tfrac12 \, \log(1-\varepsilon^2)\bigr)$;
the equal coordinates of this point are the minimum of the Chernoff information
subject to a given total variation distance $\varepsilon$.}
\end{center}
\end{figure}

As it is shown in Figure~\ref{Figure: achievable_regions_epsilon_100_140_180_198},
the boundaries of these regions become less curvy as $\varepsilon \uparrow 2$.

\subsection*{A Geometric Interpretation of the Minimal Chernoff Information with
a Constraint on the Variational Distance}
Consider the point in Figure~\ref{Figure: achievable_regions_epsilon_100_140_180_198}
which, in the plane of $\bigl( D(Q\|P_1), D(Q\|P_2) \bigr)$, is the intersection of the
straight line $D(Q\|P_1) = D(Q\|P_2)$ and the boundary of the convex region which is
characterized in Theorem~\ref{theorem: joint range of the relative entropies}
for an arbitrary $\varepsilon \in (0,2)$.

In view of the proof of Theorem~\ref{theorem: joint range of the relative entropies},
this intersection point satisfies $D(Q_{\alpha} \| P_1) = D(Q_{\alpha} \| P_2)$
for some $\alpha \in (0,1)$, for $P_1, P_2$ which are probability measures defined on a binary
alphabet with $|P_1-P_2| = \varepsilon$, and $Q_{\alpha}$ is given in
\eqref{Shayevitz function}. The equal coordinates of this intersection
point are therefore equal to the Chernoff information $C(P_1, P_2)$
(see \cite[Section~11.9]{Cover_Thomas}). Due to the symmetry of this
region with respect to the straight line $D(Q\|P_1) = D(Q\|P_2)$ (this follows from the symmetry
property $|P_1-P_2| = |P_2-P_1|$), the slope of the tangent line to the boundary of the
convex region at this intersection point is $s=-1$ (see
Figure~\ref{Figure: achievable_regions_epsilon_100_140_180_198}). This yields that
$\alpha = -\frac{s}{1-s} = \tfrac12$, and from Proposition~\ref{proposition: skew-symmetry property of g},
$g_{\alpha}(\varepsilon) = -\log\bigl(1- \frac14 \, \varepsilon^2 \bigr)$.
Hence, from \eqref{eq: straight lines}
with $\alpha = \tfrac12$, the equal coordinates of this intersection point are given by
\begin{align}
D(Q\|P_1) = D(Q\|P_2) = -\tfrac12 \, \log\bigl(1- \frac14 \, \varepsilon^2 \bigr).
\end{align}
Based on \cite[Proposition~2]{Sason_IT15}, this value is equal to the minimum of
the Chernoff information subject to an $\varepsilon$ separation constraints for
$P_1$ and $P_2$ in total variation distance.
We next calculate the probability measures $P_1^\star$, $P_2^\star$ and $Q^\star$
which attain this intersection point.
Eq.~\eqref{eq: optimization problem 1 for the Renyi divergence} with $\alpha = \frac12$ yields
\begin{align}
-2 \log \bigl( \sqrt{pq} + \sqrt{(1-p)(1-q)} \bigr) = -\log\bigl(1- \tfrac14 \varepsilon^2 \bigr)
\end{align}
such that $p,q \in [0,1]$ and $|p-q| = \frac{\varepsilon}{2}$.
A possible solution of this equation is $p = \frac{2+\varepsilon}{4}$ and $q = \frac{2-\varepsilon}{4}$,
so the respective probability measures $P_1^\star, P_2^\star$ which are defined on the binary alphabet
satisfy $P_1^\star(0) = \frac{2+\varepsilon}{4}$ and $P_2^\star(0) = \frac{2-\varepsilon}{4}$;
consequently, from \eqref{Shayevitz function}, $Q(0)=Q(1)=\tfrac12$ is the equiprobable
distribution on the binary alphabet.

As a byproduct of the characterization of the convex region in
Theorem~\ref{theorem: joint range of the relative entropies}, it follows that
the straight line $D(Q \| P_1) = D(Q \| P_2)$ (in the plane
of Figure~\ref{Figure: achievable_regions_epsilon_100_140_180_198})
intersects the boundary of the convex region which is specified in
Theorem~\ref{theorem: joint range of the relative entropies} at the point
whose coordinates are equal to the minimized Chernoff information subject to
the constraint $|P_1-P_2| \geq \varepsilon$. The equal coordinates of each of the 4
intersection points in Figure~\ref{Figure: achievable_regions_epsilon_100_140_180_198},
which refer to $\varepsilon = 1.00, 1.40, 1.80, 1.98$, are equal to
$-\tfrac12 \, \log\bigl(1- \frac14 \, \varepsilon^2 \bigr) = 0.144, 0.337, 0.830, 1.959$ nats,
respectively.

\section{A Performance Bound for Coded Communications via the R\'enyi Divergence}
\label{section: An Upper Bound on the ML Decoding Error Probability with the Renyi Divergence}

\subsection{New Exponential Upper Bound}
This section derives an exponential upper bound on the performance of binary
linear block codes, expressed in terms of the R\'enyi divergence.
Similarly to \cite{Hsu_Achilleas_IT08}, \cite{Hsu_Achilleas_IT10},
\cite{Jin_McEliece02}, \cite{MillerBurshtein}, \cite{Pfister_Siegel_IT03},
\cite[Section~3.B]{Sason_Urbanke03}, \cite{Shamai_Sason_IT02}, \cite{Shulman_Feder99}
and \cite{Twitto_IT07}, the upper bound in the next theorem quantifies the
degradation in the performance of block codes under ML decoding in terms of
the deviation of their distance spectra from the binomially distributed (average)
distance spectrum of the capacity-achieving ensemble of random block codes.

\vspace*{0.1cm}
\begin{theorem}
Consider a binary linear block code of length $N$ and rate
$R=\frac{\log(M)}{N}$ where $M$ designates the number of codewords.
Let $S_0=0$ and, for $l \in \{1, \ldots, N\}$, let $S_l$ be the number
of non-zero codewords of Hamming weight $l$. Assume that the transmission
of the code takes place over a memoryless, binary-input and output-symmetric
channel. Then, the block error probability under ML decoding satisfies
\begin{equation}
P_{\text{e}} = P_{\text{e}|0} <
\exp \left( -N \, \sup_{r \geq 1} \, \max_{0 \leq \rho' \leq \frac{1}{r}} \,
\left[E_0\left(\rho', \underline{q}=\bigl(\tfrac12, \tfrac12\bigr)\right) - \rho'
\left(rR + \frac{D_s(P_N \| Q_N)}{N}\right) \right] \right)
\label{eq: optimized error exponent}
\end{equation}
where $s \triangleq s(r) = \frac{r}{r-1}$ for $r \geq 1$
(with the convention that $s = \infty$ for $r=1$),
$Q_N$ is the binomial distribution with parameter $\tfrac12$ and $N$
independent trials (i.e., $Q_N(l) = 2^{-N} {N \choose l}$ for
$l \in \{0, 1, \ldots, N\}$), $P_N$ is the PMF defined by
$P_N(l) = \frac{S_l}{M-1}$ for $l \in \{0, \ldots, N\}$,
$D_s(\cdot \| \cdot)$ is the R\'{e}nyi divergence of order $s$ (i.e.,
$D_s(P\|Q) = \frac{1}{s-1} \, \log \left(\sum_x P(x)^s Q(x)^{1-s}\right)$
where $s>1$ here), and $E_0(\rho, \underline{q})$ designates the Gallager
random coding error exponent in \cite[Eq.~(5.6.14)]{Gallager_book1968}.
\label{theorem: channel coding theorem with the Renyi divergence}
\end{theorem}

\vspace*{0.1cm}
Before proving Theorem~\ref{theorem: channel coding theorem with the Renyi divergence},
we relate this exponential bound to previously reported bounds.

\vspace*{0.1cm}
\begin{remark}
\label{remark: generalization of the Shulman-Feder bound}
Note that the loosening of the bound by taking $r=1$ and, respectively,
$s = \infty$ gives the upper bound
\begin{align*}
P_{\text{e}} = P_{\text{e}|0} & < \exp \left( -N \, \max_{0 \leq \rho' \leq 1}
\left[E_0\left(\rho', \underline{q}=\bigl(\tfrac12, \tfrac12\bigr)\right) - \rho'
\left(R + \frac{D_{\infty}(P_N \| Q_N)}{N}\right) \right]
\right) \\[0.1cm]
& \stackrel{(\text{a})}{=} \exp \left( -N \, E_{\text{r}}\left(R +
\frac{D_{\infty}(P_N \| Q_N)}{N}\right) \right) \nonumber \\[0.1cm]
& \stackrel{(\text{b})}{=} \exp \left( -N \, E_{\text{r}}\left(R + \frac{1}{N} \,
\log \max_{0 \leq l \leq N} \frac{P_N(l)}{Q_N(l)}
\right) \right) \\[0.1cm]
& \stackrel{(\text{c})}{=} \exp \left( -N \, E_{\text{r}}\left(R + \frac{1}{N} \,
\log \max_{0 \leq l \leq N} \frac{S_l}{e^{-N(\log 2-R)} {N\choose{l}}}
\right) \right)
\end{align*}
which coincides with the Shulman-Feder bound \cite{Shulman_Feder99}. Equality~(a)
follows from the definition of the Gallager random coding exponent $E_{\text{r}}(R)$
in \cite[Eq.~(5.6.16)]{Gallager_book1968} where the symmetric input distribution
$\underline{q} = \bigl( \tfrac12, \tfrac12 \bigr)$ is the optimal input
distribution for any memoryless, binary-input output-symmetric channel, equality~(b)
follows from the expression of the R\'{e}nyi divergence of order infinity  (see, e.g.,
\cite[Theorem~6]{ErvenH14}), and equality~(c) follows from the definition of
the PMFs $P_N$ and $Q_N$ in Theorem~\ref{theorem: channel coding theorem with the Renyi divergence}.
\end{remark}

\vspace*{0.1cm}
\begin{remark}
The proof of Theorem~\ref{theorem: channel coding theorem with the Renyi divergence}
is based on the framework of the Gallager bounds in
\cite[Chapter~4]{Sason_Shamai_FnT} and \cite{Shamai_Sason_IT02}.
Specifically, it has an overlap with \cite[Appendix~A]{Shamai_Sason_IT02}.
Unlike the analysis in \cite[Appendix~A]{Shamai_Sason_IT02}, working with the R\'{e}nyi
divergence of order $s \geq 1$, instead of the relative entropy as a lower bound
(see \cite[Eq.~(A19)]{Shamai_Sason_IT02}) reveals a need for an optimization of the error
exponent, which leads to the error exponent in
Theorem~\ref{theorem: channel coding theorem with the Renyi divergence}. Namely, if the
value of $r \geq 1$ is increased then the value of $s = \frac{r}{r-1} \geq 1$ is decreased,
and therefore $D_s(P_N \| Q_N)$ is also decreased (unless it is zero, see \cite[Theorem~3]{ErvenH14};
note that $P_N$ and $Q_N$ do not depend on the parameters $r$ and $s$, so they stay un-affected by
varying the values of these parameters). The maximization of the error exponent in
Theorem~\ref{theorem: channel coding theorem with the Renyi divergence} aims at finding
a proper balance between the two summands $r R$ and $\frac{D_s(P_N \| Q_N)}{N}$ on
the right-hand side of \eqref{eq: optimized error exponent}, while also performing an optimization
over the second dependent variable $\rho' \in \bigl[0, \frac{1}{r} \bigr]$.
\end{remark}

\par
We proceed now with the proof of Theorem~\ref{theorem: channel coding theorem with the Renyi divergence}.

\begin{proof}
The proof of Theorem~\ref{theorem: channel coding theorem with the Renyi divergence}
is based on the framework of the Gallager bounds in
\cite[Chapter~4]{Sason_Shamai_FnT} and \cite{Shamai_Sason_IT02}.
Specifically, it relies on \cite[Appendix~A]{Shamai_Sason_IT02}.
We explain in the following how our proof differs from
the analysis in \cite[Appendix~A]{Shamai_Sason_IT02}.
From \cite[Eq.~(A17)]{Shamai_Sason_IT02},
we have that for every $\rho' \in \bigl[0, \frac{1}{r}\bigr]$
\begin{align}
P_{\text{e}|0} & < M^{\rho' r} \, \exp\left(-N \, E_0 \Bigl(\rho' , \underline{q} =
\Bigl(\frac{1}{2}, \frac{1}{2} \Bigr) \Bigr)\right) \, \left(\sum_{l=0}^N\, Q_N(l) \,
\left(\frac{P_N(l)}{Q_N(l)}\right)^s\right)^{\frac{r \rho'}{s}}.
\label{eq: bound with new parameters of rho and lambda}
\end{align}

From this point, we deviate from the analysis in \cite[Appendix~A]{Shamai_Sason_IT02}.
Since $\frac{1}{r}+\frac{1}{s}=1$ where $r, s \geq 1$,
we have
\begin{align}
& \left(\sum_{l=0}^N\, Q_N(l) \, \left(\frac{P_N(l)}{Q_N(l)}\right)^s\right)^{\frac{r \rho'}{s}} \nonumber \\[0.1cm]
& = \exp \left(\frac{r \rho'}{s} \cdot \log \left( \sum_{l=0}^N\, P_N(l)^s \, Q_N(l)^{1-s} \right) \right) \nonumber \\[0.1cm]
& = \exp \left(\frac{\rho'}{s-1} \cdot \log \left( \sum_{l=0}^N\, P_N(l)^s \, Q_N(l)^{1-s} \right) \right) \nonumber \\
& = \exp\bigl(\rho' D_s(P_N \| Q_N)\bigr) \label{eq: term with Renyi divergence}
\end{align}
where $D_s(P_N \| Q_N)$ is the R\'{e}nyi divergence of order $s$ from $P_N$ to $Q_N$.
This enables to refer to the R\'{e}nyi divergence of order $s \geq 1$, instead of lower
bounding this quantity by the relative entropy, and consequently loosening the
bound (see \cite[Eq.~(A19)]{Shamai_Sason_IT02}). Note that since the R\'{e}nyi divergence
is monotonically increasing in its order (see, e.g., \cite[Theorem~3]{ErvenH14}) and the
R\'{e}nyi divergence of order~1 is particularized to the relative entropy,
the inequality $D_s(P_N \| Q_N) \geq D(P_N \| Q_N)$ holds.
The combination of \eqref{eq: bound with new parameters of rho and lambda} and
\eqref{eq: term with Renyi divergence} gives
\begin{align}
P_{\text{e}|0} & < \exp(N R \rho' r) \, \exp\left(-N \, E_0 \Bigl(\rho' , \underline{q} =
\bigl(\tfrac12, \tfrac12 \bigr) \Bigr)\right) \, \exp\bigl(\rho' D_s(P_N \| Q_N)\bigr)
\nonumber \\[0.2cm]
& = \exp \left(-N \left[ E_0 \Bigl(\rho', \underline{q} =
\bigl(\tfrac12, \tfrac12 \bigr) \Bigr) - \rho' \biggl( r R + \frac{D_s(P_N \| Q_N)}{N}
\biggr) \right] \right),
\quad 0 \leq \rho' \leq \frac{1}{r}.
\label{eq: new upper bound on the decoding error probability}
\end{align}
A maximization of the error exponent in
\eqref{eq: new upper bound on the decoding error probability}
with respect to the parameters $r \geq 1$ and $\rho' \in [0, \frac{1}{r}]$
(recall that $s = s(r) = \frac{r}{r-1} > 1$) gives
the upper bound in \eqref{eq: optimized error exponent}.
\end{proof}

\subsection{Application of Theorem~\ref{theorem: channel coding theorem with the Renyi divergence}}
An efficient use of Theorem~\ref{theorem: channel coding theorem with the Renyi divergence} for
the performance evaluation of binary linear block codes (or coee ensembles) is suggested in the
following by borrowing a concept of bounding from \cite{MillerBurshtein},
which has been further studied, e.g., in \cite{Sason_Shamai_FnT}, \cite{Sason_Urbanke03},
\cite{Twitto_IT07}, and combining it with the new bound in
Theorem~\ref{theorem: channel coding theorem with the Renyi divergence}.
In order to utilize the Shulman-Feder bound for binary linear block codes
in a clever way, it has been suggested in \cite{MillerBurshtein} to partition
the binary linear block code $\mathcal{C}$ into two subcodes
$\mathcal{C}_1$ and $\mathcal{C}_2$ where $\mathcal{C}_1 \cup \mathcal{C}_2 = \mathcal{C}$
and $\mathcal{C}_1 \cap \mathcal{C}_2 = \{0\}$ is the all-zero codeword. The
first subcode $\mathcal{C}_1$ contains the all-zero codeword and all the codewords
of $\mathcal{C}$ whose Hamming weights $l$ belong to a subset
$\mathcal{L} \subseteq \{1,2,...,N\}$, while
$\mathcal{C}_2$ contains the other codewords of $\mathcal{C}$ which have Hamming
weights of $l \in \mathcal{L}^{\mathrm{c}} \triangleq \{1,2,...,N\}\setminus
\mathcal{L}$, together with the all-zero codeword. From the symmetry of the
channel, $P_{\text{e}}=P_{\text{e}|0} \leq
P_{\text{e}|0}(\mathcal{C}_1)+P_{\text{e}|0}(\mathcal{C}_2)$
where $P_{\text{e}|0}(\mathcal{C}_1)$ and
$P_{\text{e}|0}(\mathcal{C}_2)$ designate the conditional ML
decoding error probabilities of $\mathcal{C}_1$ and
$\mathcal{C}_2$, respectively, given that the all-zero codeword is
transmitted. Note that although the code $\mathcal{C}$ is
linear, its two subcodes $\mathcal{C}_1$ and $\mathcal{C}_2$ are in
general non-linear. One can rely on different upper bounds
on the conditional error probabilities
$P_{\text{e}|0}(\mathcal{C}_1)$ and
$P_{\text{e}|0}(\mathcal{C}_2)$, i.e., we may bound
$P_{\text{e}|0}(\mathcal{C}_1)$ by invoking
Theorem~\ref{theorem: channel coding theorem with the Renyi divergence},
due to its tightening of the Shulman-Feder bound (see
Remark~\ref{remark: generalization of the Shulman-Feder bound}),
and also rely on an alternative approach for obtaining an upper bound on
$P_{\text{e}|0}(\mathcal{C}_2)$ (e.g., it is possible to rely on the union bound
with respect to the fixed composition codes of the subcode
$\mathcal{C}_2$). The idea behind this partitioning is to
include in the subcode $\mathcal{C}_1$ the codewords of all the Hamming weights
whose distance spectrum is close enough to the binomial distribution $Q_N$ (see
Theorem~\ref{theorem: channel coding theorem with the Renyi divergence}) in
the sense that the additional term $\frac{D_s(P_N \| Q_N)}{N}$ in the exponent of
\eqref{eq: optimized error exponent} has a marginal effect on the conditional
ML decoding error probability of the subcode $\mathcal{C}_1$.

Theorem~\ref{theorem: channel coding theorem with the Renyi divergence}
can be applied as well to {\em ensembles} of binary linear block codes. The verify this claim,
let $\mathcal{C}$ be an ensemble of binary linear block codes.
The proof of Theorem~\ref{theorem: channel coding theorem with the Renyi divergence}
follows from the Duman and Salehi bounding technique \cite{Shamai_Sason_IT02} which leads
to the derivation of \cite[Eq.~(A.11)]{Shamai_Sason_IT02}. By taking the expectation on the RHS
of \cite[Eq.~(A.11)]{Shamai_Sason_IT02} with respect to the code ensemble $\mathcal{C}$ and invoking
Jensen's inequality, the same bound holds while $S_l$, as it is defined in
Theorem~\ref{theorem: channel coding theorem with the Renyi divergence}, is replaced
by the expectation $\overline{S_l} \triangleq \expectation_{\mathcal{C}}\bigl[S_l\bigr]$ with
respect to the code ensemble $\mathcal{C}$. This enables to replace $P_N$ on the RHS of
\eqref{eq: optimized error exponent} with $\overline{P}_N$ where  $$\overline{P}_N(l) \triangleq \frac{\expectation_{\mathcal{C}}\bigl[S_l\bigr]}{M-1}, \quad \forall \, l \in \{0, \ldots, N\},$$
which therefore justifies the generalization of
Theorem~\ref{theorem: channel coding theorem with the Renyi divergence} to code ensembles
of binary linear block codes.

As it is exemplified in Section~\ref{subsection: turbo},
Theorem~\ref{theorem: channel coding theorem with the Renyi divergence} can be efficiently
applied to ensembles of turbo-like codes in the same way that it was demonstrated to be efficient
in \cite{Twitto_IT07}. Similarly to Theorem~\ref{theorem: channel coding theorem with the Renyi divergence},
the bound in \cite[Theorem~3.1]{Twitto_IT07} forms another refinement of the Shulman-Feder bound, and the
novelty in the former bound is the obtained tightening of the Shulman-Feder bound via the use of the
R\'{e}nyi divergence.

\subsection{An Example: Performance Bounds for an Ensemble of Turbo-Block Codes}
\label{subsection: turbo}
We conclude this section by an example which applies this bounding technique to the
ensemble of uniformly interleaved turbo codes whose two component codes are chosen
uniformly at random from the ensemble of (1072, 1000) binary systematic linear block
codes. The transmission of these
codes takes place over an additive white Gaussian noise (AWGN) channel, and the codes
are BPSK modulated and coherently detected. The calculation of the average distance
spectrum of this ensemble has been performed in \cite[Section~5.D]{Twitto_IT07}, which
is required for the calculation of the upper bound in \eqref{eq: optimized error exponent}
where the PD $P_N$ is replaced by its expected value over the ensemble
(i.e., the normalization of the average distance spectrum by the number of codewords, as it
is defined in Theorem~\ref{theorem: channel coding theorem with the Renyi divergence}).
In the following, two upper bounds on the block error probability are compared under
ML decoding: the first one is the tangential-sphere bound (TSB) of Herzberg and Poltyrev
(see \cite{HerzbergP}, \cite{Poltyrev_IT1994}, \cite[Section~3.2.1]{Sason_Shamai_FnT}),
and the second bound follows from the suggested combination of the union bound and
Theorem~\ref{theorem: channel coding theorem with the Renyi divergence}.
Note that an optimal partitioning has been performed, in a way which is conceptually
similar to \cite[Algorithm~1]{Twitto_IT07}, for obtaining the tightest bound
which is obtained by combining the union bound and
Theorem~\ref{theorem: channel coding theorem with the Renyi divergence}.

A comparison of the two bounds shows an advantage of the latter combined bound over the TSB
in a similar way to \cite[upper plot of Fig.~8]{Twitto_IT07} (e.g., providing a gain of about
0.2~dB over the TSB for a block error probability of $10^{-3}$). Note that the Shulman-Feder bound
is rather loose in this case due to the significant deviation of the ensemble distance spectrum
from the binomial distribution at low and high Hamming weights. Furthermore, we note that the
advantage of the proposed bound over the TSB in this example is consistent with
the analysis in \cite{Poltyrev_IT1994} and \cite{Twitto_Sason_IT07}, demonstrating a gap between
the random coding error exponent of Gallager and the corresponding
error exponents that follow from the TSB and some of its improved versions. Recall
that the random coding error exponent of Gallager achieves the channel capacity,
whereas the random coding error exponent that follows from the TSB (or some of its
improved variants) does not achieve the capacity of a binary-input AWGN channel for
BPSK modulated fully random block codes, where the gap to capacity is especially
pronounced for high coding rates. In this example, the rate of the ensemble
is 0.8741~bits per channel use.

\appendices

\section{Proofs of Lemmas~\ref{lemma: limit when the TV distance tends to 1}
and~\ref{lemma: restriction of the minimization to 2-element probability distributions}}
\renewcommand{\thelemma}{I.\arabic{lemma}}
\renewcommand{\theequation}{I.\arabic{equation}}
\setcounter{equation}{0}
\setcounter{lemma}{0}
\label{Appendix: proofs of Lemma 1 and 2}

\subsection{Proof of Lemma~\ref{lemma: limit when the TV distance tends to 1}}
\label{Appendix: proof of Lemma 1}

For $\alpha = \tfrac12$, $D_{\frac{1}{2}}(P\|Q) = -2 \log Z(P,Q)$
where $Z(P,Q) \triangleq \sum_x \sqrt{P(x) Q(x)}$ denotes the Bhattacharyya
coefficient between the two PDs $P, Q$. We have
\begin{align} \label{eq: D12 vs TV}
D_{\frac{1}{2}}(P\|Q) \geq -\log\left(1- \tfrac14 \, \varepsilon^2\right)
\end{align}
where $|P-Q| = \varepsilon$ (see, e.g., \cite[Proposition~1]{Sason_IT15};
inequality \eqref{eq: D12 vs TV} is known in quantum information theory
with respect to the relation between the trace distance and fidelity
\cite[Section~9.3]{Wilde}). Hence, \eqref{eq: D12 vs TV} implies that
\eqref{eq: limit when the TV distance tends to 1} holds for $\alpha=\tfrac12$.
Since $D_{\alpha}(P\|Q)$ is monotonically increasing in its order $\alpha$
(see \cite[Theorem~3]{ErvenH14}), it follows that
\eqref{eq: limit when the TV distance tends to 1} also holds for
$\alpha \geq \tfrac12$. Finally, due to the skew-symmetry property of $D_{\alpha}$
(see \cite[Proposition~2]{ErvenH14}) where
$D_{\alpha}(P\|Q) = \left(\tfrac{\alpha}{1-\alpha}\right) D_{1-\alpha}(Q\|P)$
for $\alpha \in (0,1)$, and since the total variation distance is a
symmetric measure and $\tfrac{\alpha}{1-\alpha} > 0$ for $\alpha \in (0,1)$,
the satisfiability of \eqref{eq: limit when the TV distance tends to 1} for
$\alpha \in (\tfrac12, 1)$ yields that it also holds for
$\alpha \in (0, \tfrac12)$.

\subsection{Proof of Lemma~\ref{lemma: restriction of the minimization to 2-element probability distributions}}
\label{Appendix: proof of Lemma 2}

Let $P_1 \ll P_2$ be probability measures which are defined on a common measurable space $(\set{A}, \mathscr{F})$.
Denote by $\phi \colon \mathcal{A} \rightarrow \{1, 2\}$ the mapping given by
\begin{align*}
\phi(x) = \left\{
\begin{array}{ll}
1, & \quad \mbox{if $\frac{\mathrm{d}P_1}{\mathrm{d}P_2}(x) \geq 1$,} \\
2, & \quad \mbox{if $\frac{\mathrm{d}P_1}{\mathrm{d}P_2}(x) < 1$}
\end{array}
\right.
\end{align*}
and let $Q_i$, for $i \in \{1, 2\}$, be given by
\begin{equation}
Q_i(j) \triangleq \int_{\{x \in \set{A} \colon \phi(x)=j\}}
\mathrm{d}P_i(x), \quad \forall \, i, j \in \{1, 2\}.
\label{eq: Q1, Q2}
\end{equation}
Consequently, we have
\begin{align} \label{eq: same TV}
|P_1-P_2| & = \int_{\set{A}} \left|\frac{\mathrm{d}P_1}{\mathrm{d}P_2}(x) - 1\right| \, \mathrm{d}P_2(x) \nonumber \\[0.1cm]
& = \int_{\{x \in \set{A} \colon \phi(x)=1\}} \left(\frac{\mathrm{d}P_1}{\mathrm{d}P_2}(x) - 1\right) \, \mathrm{d}P_2(x) +
\int_{\{x \in \set{A} \colon \phi(x)=2\}} \left(1 - \frac{\mathrm{d}P_1}{\mathrm{d}P_2}(x)\right) \, \mathrm{d}P_2(x) \nonumber \\[0.1cm]
& = \bigl( Q_1(1) - Q_2(1) \bigr) + \bigl( Q_2(2) - Q_1(2) \bigr) \nonumber \\
& = \sum_{j \in \{1, 2\}} \bigl|Q_1(j)-Q_2(j)\bigr| \nonumber \\
& = |Q_1-Q_2|.
\end{align}
From the data processing theorem for the R\'{e}nyi divergence
(see \cite[Theorem~9]{ErvenH14}),
\begin{equation}
D_{\alpha}(P_1 \| P_2) \geq D_{\alpha}(Q_1 \| Q_2)
\label{eq: data processing inequality}
\end{equation}
where $Q_1$ and $Q_2$ are the probability measures which are defined on the binary alphabet (see
\eqref{eq: Q1, Q2}). The lemma follows by combining \eqref{eq: same TV} and
\eqref{eq: data processing inequality}.

\section{Proof of Proposition~\ref{proposition: skew-symmetry property of g}}
\renewcommand{\thelemma}{II.\arabic{lemma}}
\renewcommand{\theequation}{II.\arabic{equation}}
\setcounter{equation}{0}
\setcounter{lemma}{0}
\label{Appendix: proof of proposition on skew-symmetry of g}

Eq.~\eqref{eq: g for alpha=1/2} follows from the equality
$D_{\frac{1}{2}}(P\|Q) = -2 \log Z(P,Q)$ where $Z(P,Q)$ is
the Bhattacharyya coefficient between $P, Q$,
and since (see \cite[Proposition~1]{Sason_IT15})
$$\max_{P, Q \colon |P-Q| =\varepsilon} Z(P,Q) =
\sqrt{1- \tfrac14 \, \varepsilon^2}, \quad \forall \, \varepsilon \in [0,2).$$

To prove \eqref{eq: g for alpha=2}, note that
$D_2(P_1 \| P_2) = \log\bigl(1+\chi^2(P_1,P_2)\bigr)$
where $$\chi^2(P_1 \| P_2) \triangleq
\int \left(\frac{\mathrm{d}P_1}{\mathrm{d}P_2} - 1 \right)^2 \, \mathrm{d}P_2$$
denotes the $\chi^2$-divergence between the probability measures
$P_1$ and $P_2$ (which is the Hellinger divergence of order~2).
One can derive a closed-form expression for $g_2$ by relying on
the closed-form solution of a minimization of the $\chi^2$-divergence
$\chi^2(P_1 \| P_2)$ subject to the constraint $|P_1 - P_2| = \varepsilon \in [0,2)$,
which is given by (see \cite[Eq.~(58)]{ReidW11})
\begin{align*}
\min_{P_1, P_2 \colon |P_1-P_2| = \varepsilon}
\chi^2(P_1 \| P_2) = \left\{
\begin{array}{ll}
\varepsilon^2, & \quad \mbox{if $\varepsilon \in [0,1]$,} \\[0.1cm]
\frac{\varepsilon}{2-\varepsilon}, & \quad
\mbox{if $\varepsilon \in (1,2).$}
\end{array}
\right.
\end{align*}

Eq.~\eqref{eq: skew-symmetry property of g} follows from the
skew-symmetry property of the R\'{e}nyi divergence
\cite[Proposition~2]{ErvenH14}.

The lower bound on $g_{\alpha}$ in \eqref{eq: lower bound on g}
follows from
\eqref{eq: optimization problem 1 for the Renyi divergence},
which implies that for $\alpha \in (0,1)$ and $\varepsilon \in [0,2)$
\begin{equation}
g_{\alpha}(\varepsilon) =
\frac{\log \Bigl( \max_{p, q \in [0,1] \colon |p-q| \geq \frac{\varepsilon}{2}}
\left(p^{\alpha} q^{1-\alpha} + (1-p)^{\alpha}
(1-q)^{1-\alpha} \right) \Bigr)}{\alpha-1}
\label{eq: an equivalent expression for g}
\end{equation}
and, we have
\begin{align}
0 & \leq \max_{p, q \in [0,1] \colon |p-q| \geq \frac{\varepsilon}{2}}
\bigl(p^{\alpha} q^{1-\alpha} + (1-p)^{\alpha}
(1-q)^{1-\alpha} \bigr) \nonumber \\[0.1cm]
& \leq \max_{p, q \in [0,1] \colon |p-q| \geq \frac{\varepsilon}{2}}
p^{\alpha} q^{1-\alpha} +
\max_{p, q \in [0,1] \colon |p-q| \geq \frac{\varepsilon}{2}}
(1-p)^{\alpha} (1-q)^{1-\alpha} \nonumber \\[0.1cm]
& = 2 \max_{p, q \in [0,1] \colon |p-q| \geq \frac{\varepsilon}{2}}
p^{\alpha} q^{1-\alpha} \nonumber \\[0.1cm]
& = 2 \, \max \left\{ \bigl(1-\tfrac12 \, \varepsilon \bigr)^{\alpha}, \,
\bigl(1-\tfrac12 \, \varepsilon \bigr)^{1-\alpha} \right\}.
\label{eq: upper bound on the max term for g}
\end{align}
The lower bound on $g_{\alpha}$ in \eqref{eq: lower bound on g}
follows from the combination of
\eqref{eq: an equivalent expression for g} and
\eqref{eq: upper bound on the max term for g}.

\section{Proof of Lemma~\ref{lemma: monotonicity of f}}
\renewcommand{\thelemma}{III.\arabic{lemma}}
\renewcommand{\theequation}{III.\arabic{equation}}
\setcounter{equation}{0}
\setcounter{lemma}{0}
\label{Appendix: proof of lemma on monotonicity}

For $\alpha \in (0,1)$ and $\varepsilon' \in (0,1)$, we have
\begin{align*}
& \lim_{q \rightarrow 0^+} \left(1+\frac{\varepsilon'}{q}\right)^{\alpha-1} = 0,
\qquad \lim_{q \rightarrow 0^+} \left(1+\frac{\varepsilon'}{q}\right)^\alpha = +\infty, \\
& \Longrightarrow \lim_{q \rightarrow 0^+} f_{\alpha, \varepsilon'}(q)
= \lim_{q \rightarrow 0^+} \frac{(1-\varepsilon')^{\alpha-1}}{\left(1+
\frac{\varepsilon'}{q}\right)^\alpha -(1-\varepsilon')^\alpha} = 0,
\end{align*}
and
\begin{align*}
& \lim_{q \rightarrow (1-\varepsilon')^{-}}
\left(1-\frac{\varepsilon'}{1-q}\right)^{\alpha-1} = +\infty, \qquad
\lim_{q \rightarrow (1-\varepsilon')^{-}}
\left(1-\frac{\varepsilon'}{1-q}\right)^\alpha = 0, \\
& \Longrightarrow \lim_{q \rightarrow (1-\varepsilon')^{-}}
f_{\alpha, \varepsilon'}(q) = \lim_{q \rightarrow (1-\varepsilon')^{-}}
\frac{\left(1-\frac{\varepsilon'}{1-q}\right)^{\alpha-1}
- (1-\varepsilon')^{1-\alpha}}{(1-\varepsilon')^{-\alpha}
-\left(1-\frac{\varepsilon'}{1-q}\right)^\alpha} = +\infty.
\end{align*}
This proves the two limits in \eqref{eq: limits of f at the 2 endpoints of the interval}.

We prove in the following that $f_{\alpha, \varepsilon'}(\cdot)$ is
strictly increasing on the interval
$\bigl[\frac{1-\varepsilon'}{2}, 1-\varepsilon')$, and we also prove
later in this appendix that this function is monotonically increasing
on the interval $\bigl(0, \frac{1-\varepsilon'}{2}\bigr]$. These
two parts of the proof yield that $f_{\alpha, \varepsilon'}(\cdot)$
is strictly monotonically increasing on the interval $(0, 1-\varepsilon')$.
The positivity of $f_{\alpha, \varepsilon'}$ on $(0, 1-\varepsilon')$
follows from the first limit
in \eqref{eq: limits of f at the 2 endpoints of the interval}, jointly with
the monotonicity of this function which is proved in the following.

For a proof that $f_{\alpha, \varepsilon'}(\cdot)$ is strictly
monotonically increasing on $\bigl[\frac{1-\varepsilon'}{2}, 1-\varepsilon')$,
this function (see \eqref{eq: f of the 2 parameters alpha and epsilon})
is expressed as follows:
\begin{align}
f_{\alpha, \varepsilon'}(q) 
& = \left(1+\frac{\varepsilon'}{q} \right)^{-1}
u_{\alpha}\bigl(z_{\varepsilon'}(q)\bigr)
\label{eq: 1st way for rewritting f}
\end{align}
where
\begin{align}
& z_{\varepsilon'}(q) \triangleq
\frac{1-\frac{\varepsilon'}{1-q}}{1+\frac{\varepsilon'}{q}} \, ,
\label{eq: z of epsilon and q} \\
& u_{\alpha}(t) \triangleq \left\{
\begin{array}{ll}
\frac{t^{\alpha-1}-1}{1-t^\alpha}, & \quad \mbox{if $t \in
(0,1) \cup (1,\infty)$,} \\[0.2cm]
\tfrac{1-\alpha}{\alpha}, & \quad \mbox{if $t=1$.}
\end{array}
\right.
\label{eq: function u}
\end{align}
Note that $u_{\alpha}$ in \eqref{eq: function u} was
defined to be continuous at $t=1$. In order to proceed, we
need the following two lemmas:

\begin{lemma}
Let $\varepsilon' \in (0,1)$. The function $z_{\varepsilon'}$
in \eqref{eq: z of epsilon and q} is strictly monotonically increasing
on $\bigl(0, \frac{1-\varepsilon'}{2} \bigr]$, and it is strictly
monotonically decreasing on $\bigl[\frac{1-\varepsilon'}{2}, 1-\varepsilon')$.
This function is also positive on $(0, 1-\varepsilon')$.
\label{lemma: properties of the function z}
\end{lemma}
\begin{proof}
$z_{\varepsilon'}(q)>0$ for $q \in (0, 1-\varepsilon')$ since
$1-\frac{\varepsilon'}{1-q}>0$, and $1+\frac{\varepsilon'}{q} > 0$.
In order to prove the monotonicity properties of $z_{\varepsilon'}$,
note that its derivative satisfies the equality
\begin{equation}
\frac{\mathrm{d}}{\mathrm{d}q} \; z_{\varepsilon'}(q) = \varepsilon' \, z_{\varepsilon'}(q) \,
\left( \frac{1}{q(\varepsilon'+q)} - \frac{1}{(1-q)(1-\varepsilon'-q)} \right)
\label{eq: derivative of z}
\end{equation}
which is derived by taking logarithms on both sides of
\eqref{eq: z of epsilon and q}, followed by their differentiation. By
setting the derivative of $z_{\varepsilon'}(q)$ (with respect to $q$)
to zero, we have $q=\frac{1-\varepsilon'}{2}$. Since
$z_{\varepsilon'}(q)>0$ for $q \in (0, 1-\varepsilon')$, it follows from
\eqref{eq: derivative of z} that $z'_{\varepsilon'}(q) > 0$ for
$q \in \bigl(0, \frac{1-\varepsilon'}{2} \bigr)$, and $z'_{\varepsilon'}(q) < 0$
for $q \in \bigl(\frac{1-\varepsilon'}{2}, 1-\varepsilon')$. Hence,
$z_{\varepsilon'}$ is strictly monotonically increasing
on $\bigl(0, \frac{1-\varepsilon'}{2} \bigr]$, and it is strictly
monotonically decreasing on $\bigl[\frac{1-\varepsilon'}{2}, 1-\varepsilon')$.
\end{proof}

\begin{lemma}
Let $\alpha \in (0,1)$.
The function $u_{\alpha}$ in \eqref{eq: function u} is strictly monotonically
decreasing and positive on $(0, \infty)$.
\label{lemma: properties of the function u}
\end{lemma}
\begin{proof}
Differentiation of $u_{\alpha}$ in \eqref{eq: function u} gives that for $t>0$
\begin{equation}
u'_{\alpha}(t) = \frac{t^{\alpha-2} \, (t^\alpha - \alpha t + \alpha-1)}{(t^\alpha-1)^2}.
\label{eq: derivative of u}
\end{equation}
Note that
$\frac{d}{dt} \left(t^\alpha - \alpha t + \alpha-1 \right) = \alpha (t^{\alpha-1}-1)$,
so the derivative is zero at $t=1$, it is positive if $t \in (0,1)$,
and it is negative if $t \in (1, \infty)$. This implies that
$t^{\alpha}-\alpha t + \alpha-1 \leq 0$ for every $t \in (0, \infty)$, and it is
satisfied with equality if and only if $t=1$. From \eqref{eq: derivative of u}, it
follows that $u_{\alpha}$ is strictly monotonically decreasing on $(0, \infty)$. Since
$\lim_{t \rightarrow \infty} u_{\alpha}(t) = 0$ (see \eqref{eq: function u}) and
$u_{\alpha}$ is strictly monotonically decreasing on $(0, \infty)$ then it is positive
on this interval.
\end{proof}

From Lemmas~\ref{lemma: properties of the function z}
and~\ref{lemma: properties of the function u}, it follows
that $z_{\varepsilon'}$ is strictly monotonically decreasing
and positive on $\bigl[\frac{1-\varepsilon'}{2}, 1-\varepsilon' \bigr)$,
and $u_{\alpha}$ is strictly monotonically decreasing and positive
on $(0, \infty)$. This therefore implies that the composition
$u_{\alpha} \bigl(z_{\varepsilon'}(\cdot) \bigr)$
is strictly monotonically increasing and positive on the interval
$\bigl[\frac{1-\varepsilon'}{2}, 1-\varepsilon' \bigr)$. Hence,
from \eqref{eq: 1st way for rewritting f}, since $f_{\alpha, \varepsilon'}(\cdot)$
is expressed as a product of two positive and strictly monotonically increasing
functions on $\bigl[\frac{1-\varepsilon'}{2}, 1-\varepsilon' \bigr)$, also
$f_{\alpha, \varepsilon'}$ has these properties on this interval. This
completes the first part of the proof where we show that
$f_{\alpha, \varepsilon'}(\cdot)$ is strictly monotonically increasing
and positive on $\bigl[\frac{1-\varepsilon'}{2}, 1-\varepsilon')$.

\vspace*{0.1cm}
We prove in the following that $f_{\alpha, \varepsilon'}(\cdot)$ is also
strictly monotonically increasing and positive on $\bigl(0, \frac{1-\varepsilon'}{2}]$.
For this purpose, the function $f_{\alpha, \varepsilon'}$ is expressed
in the following alternative way:
\begin{align}
f_{\alpha, \varepsilon'}(q) & = \frac{1}{1-\frac{\varepsilon'}{q-1}}
\left(\frac{1-\frac{\varepsilon'}{q-1}}{1+\frac{\varepsilon'}{q}}\right)^{\alpha}
\; \; \frac{1-\left(\frac{1+\frac{\varepsilon'}{q}}{1-\frac{\varepsilon'}{q-1}}
\right)^{\alpha-1}}{1-\left(\frac{1-\frac{\varepsilon'}{1-q}}{1+
\frac{\varepsilon'}{q}}\right)^\alpha} \nonumber \\
& = \left(1-\frac{\varepsilon'}{1-q} \right)^{-1}
r_{\alpha}\bigl(z_{\varepsilon'}(q)\bigr)
\label{eq: 2nd way for rewritting f}
\end{align}
where $z_{\varepsilon'}$ is defined in \eqref{eq: z of epsilon and q}, and
\begin{align}
r_{\alpha}(t) \triangleq \left\{
\begin{array}{ll}
\frac{t^\alpha (1-t^{1-\alpha})}{1-t^\alpha}, & \quad \mbox{if $t \in
(0, \infty) \setminus \{1\}$,} \\[0.1cm]
\tfrac{1-\alpha}{\alpha}, & \quad \mbox{if $t=1$.}
\end{array}
\right.
\label{eq: function r}
\end{align}
Note that it follows from Lemma~\ref{lemma: properties of the function z}
and \eqref{eq: z of epsilon and q} that
$$z_{\varepsilon'}(q) \leq z_{\varepsilon'}\left(\frac{1-\varepsilon'}{2}\right)
= \left(\frac{1-\varepsilon'}{1+\varepsilon'}\right)^2 < 1$$
so the composition $r_{\alpha}\bigl(z_{\varepsilon'}(\cdot)\bigr)$ in
\eqref{eq: 2nd way for rewritting f} is independent of $r_{\alpha}(1)$; the
value of $r_{\alpha}(1)$ is defined in \eqref{eq: function r} to obtain the
continuity of $r_{\alpha}$, which leads to the following lemma:

\begin{lemma}
For $\alpha \in (0,1)$, the function $r_{\alpha}$ in \eqref{eq: function r}
is strictly monotonically increasing and positive on $(0, \infty)$.
\label{lemma: properties of the function r}
\end{lemma}
\begin{proof}
A differentiation of $r_{\alpha}$ in \eqref{eq: function r} gives
\begin{equation}
r'_{\alpha}(t) = \frac{(1-\alpha)t^{\alpha} + \alpha
t^{\alpha-1} - 1}{(t^\alpha-1)^2}
\label{eq: derivative of r}
\end{equation}
so the sign of $r'_{\alpha}$ is the same as of
$(1-\alpha) t^{\alpha} + \alpha t^{\alpha-1} - 1$.
Since $\alpha \in (0,1)$, and
$$\frac{d}{dt} \bigl((1-\alpha) t^{\alpha} + \alpha t^{\alpha-1} - 1 \bigr)
= \alpha (1-\alpha) t^{\alpha-2} (t-1)$$
it follows that the last derivative is negative for $t \in (0,1)$,
zero at $t=1$, and positive for $t \in (1, \infty)$. This implies
that $t=1$ is a global minimum of the numerator of $r'_\alpha$
(see \eqref{eq: derivative of r}), so
$$ (1-\alpha) t^{\alpha} + \alpha \, t^{\alpha-1} - 1 \geq 0,
\quad \forall \, t \in (0, \infty)$$
and equality holds if and only if $t=1$. It therefore
follows from \eqref{eq: derivative of r} that $r'_{\alpha}(t) > 0$
for $t \in (0, \infty) \setminus \{1\}$, so $r_{\alpha}(\cdot)$
is strictly monotonically increasing on $(0, \infty)$.
Since $\lim_{t \rightarrow 0} r_{\alpha}(t) = 0$, the monotonicity of
$r_{\alpha}(\cdot)$ on $(0, \infty)$ yields that it is positive on
this interval.
\end{proof}

From Lemmas~\ref{lemma: properties of the function z}
and~\ref{lemma: properties of the function r}, $z_{\varepsilon'}$
is strictly monotonically increasing and positive on
$\bigl(0, \frac{1-\varepsilon'}{2}\bigr]$, and $r_{\alpha}$
is strictly monotonically increasing and positive on $(0, \infty)$.
This implies that the composition
$r_{\alpha}\bigl( z_{\varepsilon'}(\cdot) \bigr)$
is strictly monotonically increasing and positive on the interval
$\bigl(0, \frac{1-\varepsilon'}{2}\bigr]$. From
\eqref{eq: 2nd way for rewritting f}, $f_{\alpha, \varepsilon'}$
is expressed as a product of two strictly increasing and positive
functions on the interval $\bigl(0, \frac{1-\varepsilon'}{2} \bigr]$,
which implies that $f_{\alpha, \varepsilon'}(\cdot)$
also has these properties on this interval. This completes the
second part of the proof where we show that
$f_{\alpha, \varepsilon'}(\cdot)$ is strictly monotonically increasing
and positive on $\bigl(0, \frac{1-\varepsilon'}{2}\bigr]$.
The combination of the two parts of this proof
completes the proof of Lemma~\ref{lemma: monotonicity of f}.

\section{Proof of Proposition~\ref{proposition: efficient calculation of g_alpha}}
\renewcommand{\thelemma}{IV.\arabic{lemma}}
\renewcommand{\theequation}{IV.\arabic{equation}}
\setcounter{equation}{0}
\setcounter{lemma}{0}
\label{Appendix: proof of proposition on efficient calculation of g}

The proof relies on the Lagrange duality and KKT conditions, where strong duality is
first asserted by verifying the satisfiability of Slater's condition.

Let $\alpha \in (0,1)$, $\varepsilon \in (0,2)$, and $\varepsilon' = \frac{\varepsilon}{2}$.
Solving \eqref{eq: optimization problem 1 for the Renyi divergence}
is equivalent to solving the optimization problem
\begin{align}
& \text{maximize} \quad p^{\alpha} q^{1-\alpha} + (1-p)^{\alpha} (1-q)^{1-\alpha} \nonumber \\
& \text{subject to} \label{eq: equivalent optimization problem for g} \\
& \qquad \qquad \left\{
\begin{array}{ll}
p, q \in [0,1], \nonumber \\[0.1cm]
|p-q| \geq \varepsilon' \nonumber
\end{array}
\right.
\end{align}
where $p, q$ are the optimization variables. The objective
function of the optimization problem \eqref{eq: equivalent optimization problem for g}
is concave for $\alpha \in (0,1)$, so this maximization problem
is a convex optimization problem. Since the problem is also strictly feasible
at an interior point of the domain in \eqref{eq: equivalent optimization problem for g},
Slater's condition yields that strong duality holds for this optimization
problem (see \cite[Section~5.2.3]{CVX_book}).
Note that the replacement of $p,q$ with $1-p$ and $1-q$, respectively,
does not affect the value of the objective function and the satisfiability
of the constraints in \eqref{eq: equivalent optimization problem for g}.
Consequently, it can be assumed with loss of generality that $p \geq q$;
together with the inequality constraint $|p-q| \geq \varepsilon'$, it gives
that $p-q \geq \varepsilon'$. The Lagrangian of the dual problem is given by
$$L(p,q,\lambda) = p^{\alpha} q^{1-\alpha} + (1-p)^{\alpha} (1-q)^{1-\alpha}
+ \lambda (q-p+\varepsilon')$$
and the KKT conditions lead to the following set of equations:
\begin{align}
\left\{
\begin{array}{ll}
& \frac{\partial L}{\partial p} =  \alpha \bigl[p^{\alpha-1} q^{1-\alpha}
- (1-p)^{\alpha-1} (1-q)^{1-\alpha} \bigr] - \lambda = 0, \\[0.2cm]
& \frac{\partial L}{\partial q} = (1-\alpha) \bigl[ p^{\alpha} q^{-\alpha}
- (1-p)^{\alpha} (1-q)^{-\alpha} \bigr] + \lambda = 0, \\[0.2cm]
& \frac{\partial L}{\partial \lambda} = q-p+\varepsilon' = 0.
\end{array}
\right.
\label{eq: set of KKT equations}
\end{align}
Eliminating $\lambda$ from the first equation in \eqref{eq: set of KKT equations},
and substituting it into the second equation gives
\begin{equation}
(1-\alpha) \left[ \Bigl(\frac{p}{q}\Bigr)^{\alpha} - \Bigl(\frac{1-p}{1-q}\Bigr)^{\alpha} \right]
+ \alpha \left[ \Bigl(\frac{p}{q}\Bigr)^{\alpha-1} - \Bigl(\frac{1-p}{1-q}\Bigr)^{\alpha-1} \right] = 0.
\label{eq: equation 1 from 1st and 2nd KKT conditions}
\end{equation}
From the third equation of \eqref{eq: set of KKT equations},
Substituting $p=q+\varepsilon'$ into \eqref{eq: equation 1 from 1st and 2nd KKT conditions},
and re-arranging terms gives the equation
$f_{\alpha, \varepsilon'}(q) = \tfrac{1-\alpha}{\alpha},$
where $f_{\alpha, \varepsilon'}$ is the function in \eqref{eq: f of the 2 parameters alpha and epsilon}.

\section{Proof of Lemma~\ref{lemma: identity for the Renyi divergence}}
\renewcommand{\thelemma}{V.\arabic{lemma}}
\renewcommand{\theequation}{V.\arabic{equation}}
\setcounter{equation}{0}
\setcounter{lemma}{0}
\label{Appendix: proof of the lemma with the identity for the Renyi divergence}
For $\alpha \in (0, \infty) \setminus \{1\}$, the following equalities hold:
{\small \begin{align*}
& D(Q \| P_2) + \tfrac{\alpha}{1-\alpha} \cdot D(Q \| P_1) + \tfrac{1}{\alpha-1} \cdot D(Q \| Q_{\alpha}) \\[0.1cm]
& \stackrel{\text{(a)}}{=} \frac1{\alpha-1} \, \dint_{\set{A}} \mathrm{d}Q \, \log \left(\frac{\mathrm{d}Q}{\mathrm{d}P_2}\right)^{\alpha-1} +
\frac{1}{1-\alpha} \, \dint_{\set{A}} \mathrm{d}Q \, \log \left(\frac{\mathrm{d}Q}{\mathrm{d}P_1}\right)^\alpha +
\frac1{\alpha-1} \, \dint_{\set{A}} \mathrm{d}Q \, \log \left(\frac{\mathrm{d}Q}{\mathrm{d}Q_\alpha}\right) \\[0.1cm]
& \stackrel{\text{(b)}}{=} \frac{1}{\alpha-1} \, \dint_{\set{A}} \mathrm{d}Q(x) \, \log \left( \left(\frac{\mathrm{d}P_1}{\mathrm{d}Q} \, (x)\right)^{\alpha} \left(\frac{\mathrm{d}P_2}{\mathrm{d}Q} \, (x)\right)^{1-\alpha} \left(\frac{\mathrm{d}Q_\alpha}{\mathrm{d}Q} \, (x)\right)^{-1} \right) \\
& \stackrel{\text{(c)}}{=} \frac{1}{\alpha-1} \, \dint_{\set{A}} \mathrm{d}Q(x) \; \log \left( \dint_{\set{A}}{\left(\frac{\mathrm{d}P_1}{\mathrm{d}Q} \, (u)
\right)^\alpha} \, \left(\frac{\mathrm{d}P_2}{\mathrm{d}Q} \, (u)\right)^{1-\alpha} \, \mathrm{d}Q(u)\right) \\[0.1cm]
& \stackrel{\text{(d)}}{=} \frac{1}{\alpha-1} \; \log \left( \dint_{\set{A}}{\left(\frac{\mathrm{d}P_1}{\mathrm{d}Q} \, (u)
\right)^\alpha} \, \left(\frac{\mathrm{d}P_2}{\mathrm{d}Q} \, (u)\right)^{1-\alpha} \, \mathrm{d}Q(u)\right) \\[0.1cm]
& \stackrel{\text{(e)}}{=} D_{\alpha}(P_1 \| P_2)
\end{align*}}
where (a) follows from the equality
\begin{align}
\label{eq:RD2}
D_{\alpha}(P \| Q) &= \frac1{\alpha-1} \; \log \left( \int \mathrm{d}R \left(\frac{\mathrm{d}P}{\mathrm{d}R}\right)^\alpha \left(\frac{\mathrm{d}Q}{\mathrm{d}R}\right)^{1-\alpha} \right)
\end{align}
where $R$ is an arbitrary probability measure such that $P, Q \ll R$;
(b) holds since $P_1, P_2, Q$ are mutually absolutely
continuous which also yields that $Q \ll \gg Q_{\alpha}$ (in view of \eqref{eq: optimized tilted measure}),
(c) follows from \eqref{eq: optimized tilted measure}, (d) holds since $Q$ is a probability measure, and
(e) follows from \eqref{eq:RD2} (recall that $Q \ll P_1, P_2$).

\subsection*{Acknowledgment}
Sergio Verd\'{u} is acknowledged for a stimulating discussion on this work.
The author would like to acknowledge one of the three anonymous reviewers for some detailed and 
constructive editorial comments.

\end{document}